\documentclass[a4paper,UKenglish,cleveref, autoref, thm-restate]{lipics-v2021}

\pdfoutput=1 
\hideLIPIcs  

\usepackage[ruled,vlined,boxed]{algorithm2e}
\usepackage[noend]{algpseudocode}
\newlength{\commentWidth}
\let\oldnl\nl
\newcommand{\nonl}{\renewcommand{\nl}{\let\nl\oldnl}}
\SetKwInOut{Input}{Input}\SetKwInOut{Output}{Output}

\newcommand{\LLL}{\mathcal{L}}
\newcommand{\DDD}{\mathcal{D}}
\newcommand{\AAA}{\mathcal{A}}

\newcommand{\MMM}{\mathcal{M}}

\newcommand{\IN}{\texttt{IN}}
\newcommand{\OUT}{\texttt{OUT}}
\newcommand{\pred}{\texttt{pred}}
\newcommand{\succs}{\texttt{succ}}
\DeclareMathOperator{\LEX}{\ensuremath{LEX}}
\DeclareMathOperator{\Pref}{\ensuremath{Pref}}

\newcommand{\minWheeler}{{\normalfont\texttt{MinWheeler}}\xspace}

\usepackage{tikz}
\usetikzlibrary{automata, positioning, arrows, matrix, calc, arrows.meta}

\tikzset{
        ->,  
        >=stealth', 
        node distance=2cm, 
        every state/.style={thick, fill=gray!10}, 
        initial text=start, 
        mycell/.style={draw, minimum width=1cm, minimum height=0.7cm},
        }

\usepackage[normalem]{ulem}

\newtheorem{problem}{Problem}

\usepackage{mathtools}

\usepackage{booktabs}

\usepackage[most]{tcolorbox}

\usepackage{xcolor}

\newcommand{\mathbox}[2]{\colorbox{#1}{$\displaystyle #2$}}

\usepackage{standalone}

\title{On Computing Minimum Wheeler DFA From Their Language}

\author{Ruben Becker}{Ca' Foscari University of Venice, Italy}{rubensimon.becker@unive.it}{https://orcid.org/0000-0002-3495-3753}{}

\author{Davide Cenzato}{Ca' Foscari University of Venice, Italy}{davide.cenzato@unive.it}{https://orcid.org/0000-0002-0098-3620}{}

\author{Nicola Prezza}{Ca' Foscari University of Venice, Italy}{nicola.prezza@unive.it}{https://orcid.org/0000-0003-3553-4953}{}

\author{Daniel Puttini}{Ca' Foscari University of Venice, Italy}{daniel.puttini@unive.it}{https://orcid.org/0009-0006-8401-9949}{}

\authorrunning{R. Becker, D. Cenzato, N. Prezza, and D. Puttini} 

\Copyright{Ruben Becker, Davide Cenzato, Nicola Prezza, and Daniel Puttini} 



\nolinenumbers 

\EventEditors{Philip Bille, Seth Pettie, and Sabine Storandt}
\EventNoEds{3}
\EventLongTitle{34th Annual European Symposium on Algorithms (ESA 2026)}
\EventShortTitle{ESA 2026}
\EventAcronym{ESA}
\EventYear{2026}
\EventDate{August 31--September 4, 2026}
\EventLocation{L'Aquila, Italy}
\EventLogo{}
\SeriesVolume{388}
\ArticleNo{147}

\usepackage{adjustbox}
\usepackage{booktabs}
\usepackage{varwidth}
\usepackage{graphicx} 
\usepackage{tabularx}
\usepackage{float}

\newcommand{\ArxivOrCr}[2]{#1}

\funding{Funded by the European Union (ERC, REGINDEX, 101039208). Views and opinions expressed are however those of the author(s) only and do not necessarily reflect those of the European Union or the European Research Council Executive Agency. Neither the European Union nor the granting authority can be held responsible for them.}

\begin{document}

\maketitle

\begin{abstract}
Wheeler automata have recently emerged as a powerful generalization of the Burrows-Wheeler Transform, enabling optimal linear-time pattern matching on compressed labeled graphs --- a task that is otherwise computationally hard. Consequently, when an automaton recognizes a Wheeler language (i.e., it is equivalent to some Wheeler automaton), computing its minimum equivalent Wheeler DFA is a powerful indexing strategy. This problem is particularly relevant in computational pangenomics, where pangenome graphs frequently recognize Wheeler languages.

However, constructing the minimum Wheeler DFA for a Wheeler language has remained a computational bottleneck. 
The problem is known to be PSPACE-hard for nondeterministic inputs. When the input is a DFA, state-of-the-art solutions forced a compromise: they were either fast but limited to acyclic DFAs (Alanko et al., SODA 2020) or capable of handling general topologies but prohibitively slow (D'Agostino et al., TCS 2023). In this work, we bridge this gap with the first algorithm solving the problem for general DFAs in near-optimal, linearithmic output-sensitive time. 
By matching the efficiency of acyclic-only solutions while retaining full generality, our approach improves upon the previous general solution by at least a quadratic factor. We demonstrate the practical impact of our algorithm on real-world pangenome graphs; our tool achieves a processing throughput of over $10^5$ transitions per second on a standard workstation, enabling the construction of a provably optimal pattern matching data structure in such applications.
\end{abstract}

\newpage
\setcounter{page}{1}
\section{Introduction}

Pattern matching lies at the heart of computer science, but while we have mastered it on strings, much work remains to be done on labeled graphs. As Equi et al.~\cite{EquiGMT19} proved, pattern matching on general labeled graphs can likely not be solved in strongly subquadratic time; as a result, even basic pattern-matching-related tasks can quickly become prohibitively expensive depending on graph topology and size. This algorithmic bottleneck has long slowed down progress in data-intensive fields like bioinformatics, where pangenomics requires aligning millions of sequencing reads against complex graph-structured collections of genomes.

\emph{Wheeler automata}, introduced by Gagie et al.~\cite{GAGIE201767}, side-step the lower bound of Equi et al.~\cite{EquiGMT19} by imposing a specific structural constraint: informally speaking, their states can be totally sorted according to the co-lexicographic order of the strings labeling the paths reaching them. In the (simpler) case of Wheeler DFA, this means that $u<v$ in Wheeler order for any two states $u,v$ if and only if all (possibly, left-infinite) strings labeling paths entering in $u$ are co-lexicographically smaller than those labeling paths entering in $v$.
This seemingly simple ordering property unlocks a ``best-of-both-worlds'' scenario: it generalizes the celebrated Burrows-Wheeler Transform (BWT)~\cite{burrows1994block} from linear strings to complex graphs, enabling pattern matching in optimal constant time per character---just as if the graph were a simple string.  
Furthermore, Wheeler automata can be encoded in a constant number of bits per transition---again, as if they were simple strings--- and every Wheeler NFA can be converted into a Wheeler DFA of at most twice the size in polynomial time~\cite{alanko2020regular}. The survey \cite{cotumaccio_et_al:OASIcs.Manzini.12} gives an overview of the myriad virtues of Wheeler automata and languages. 

However, a critical gap remains. A regular language is called a \emph{Wheeler language} if it can be recognized by \emph{some} Wheeler NFA. But knowing a language is Wheeler is one thing; actually constructing the minimum Wheeler automaton for it is another. 
This led Alanko et al. \cite{alanko2020regular} to introduce a natural optimization problem: 
\begin{problem}\label{problem:Wheelerization}
    Given the minimum DFA $\DDD$ accepting a Wheeler language $\LLL$, find the minimum Wheeler DFA $\DDD_w$ accepting $\LLL$.
\end{problem}

We emphasize that Problem \ref{problem:Wheelerization} is fundamentally different from that of finding the minimum Wheeler DFA being equivalent to a given \emph{Wheeler} DFA, a problem admitting a linear-time solution \cite{AlankoCP22}.
Furthermore, while (i) deciding if a given NFA is Wheeler and (ii) deciding whether the \emph{language} of a given NFA is Wheeler are both NP-hard problems~\cite{GibneyT19,d2021ordering}, we can check if the language of a DFA is Wheeler in quadratic time~\cite{Becker2023TestWheelerness}. 
This places \cref{problem:Wheelerization} in a fascinating ``sweet spot'': it is tractable enough to determine whether a solution exists~\cite{Becker2023TestWheelerness}, yet complex enough that the optimal target automaton ($\DDD_w$) might be exponentially larger than the input ($\DDD$)~\cite{alanko2020regular}. 
This means that the ultimate goal for this problem is an algorithm whose running time scales favorably with the size of the input plus the solution it produces. 

\subparagraph*{State of the Art.} 

Let $n$ and $m$ denote the number of states and transitions of the input \emph{minimum} DFA $\DDD$, and let $n_w$ and $m_w$ denote the corresponding counts for the output Wheeler DFA $\DDD_w$ in Problem \ref{problem:Wheelerization} (in particular, $n-1 \le m \le m_w$ and $n_w-1 \le m_w$ hold).
Alanko et al.~\cite{alanko2020regular} showed how to solve the problem in $O(m_w \log m_w)$ time, though their solution was restricted to acyclic DFA. Subsequently, D'Agostino et al.~\cite{d2023ordering} proposed the first algorithm capable of handling arbitrary DFA. However, while their method solves the general case, it requires $O(n^3 n_w + n^2 \sigma n_w \log n_w)$ time (where $\sigma$ is the alphabet size), making it computationally prohibitive for large inputs and thus primarily of theoretical interest.

\subparagraph*{Our Contributions and Techniques.} 

In this paper, we establish the following main result:

\begin{restatable}[]{theorem}{maintheorem}
\label{thm:main theorem}
    Given a minimum DFA $\DDD$ accepting a Wheeler language $\LLL$, the minimum Wheeler DFA $\DDD_w$ for $\LLL$---having  $m_w$ transitions---can be computed in $O(m_w \log m_w)$
    time. 
\end{restatable}

As discussed below, the algorithm behind Theorem \ref{thm:main theorem} terminates if and only if the input language $\LLL(\DDD)$ is Wheeler. However, combining the above theorem with the quadratic-time algorithm of Becker et al.~\cite{Becker2023TestWheelerness} for checking whether the language of a given DFA is Wheeler, we obtain an algorithm that terminates for \emph{any input DFA} $\DDD$, either determining that $\LLL(\DDD)$ is not Wheeler, or returning the corresponding minimum WDFA $\DDD_w$. 

Theorem \ref{thm:main theorem} effectively combines the strengths of both prior approaches: we match the near-optimal complexity of Alanko et al.~\cite{alanko2020regular} while also retaining the ability to process any input DFA, as in the algorithm by D'Agostino et al.~\cite{d2023ordering} (but $\Omega(n^2)$ times faster).

\subparagraph*{Experimental results.}
As we demonstrate experimentally, 
in real-world pangenome graphs, $m_w$ is often very close to $n$; this is due to the fact that in those applications the input DFA is very sparse and often not much smaller than the output WDFA. As a result, in those applications, our optimized implementation of the algorithm behind \cref{thm:main theorem} outputs over $10^5$ transitions per second on a standard workstation. These performances are comparable with those of the only other existing tool \cite{GCSA} able to convert a pangenome graph into a (possibly not minimal) equivalent Wheeler DFA. 

\subparagraph*{Intuitive Description of our Results.}

The algorithm that allows us to derive \cref{thm:main theorem} starts by sorting co-lexicographically any spanning tree of $\DDD$ rooted in the source. As observed by Alanko et al. \cite{alanko2020regular}, this yields a spanning Wheeler subgraph $\DDD_w$ of $\DDD$.
Let $\Delta$ be a queue initialized with all transitions of $\DDD$ not included in $\DDD_w$.
While $\Delta\neq \emptyset$, the algorithm inserts those transitions into $\DDD_w$, possibly splitting its states and pushing new transitions into $\Delta$, in such a way that 
$\DDD_w$ remains a WDFA all the time and that
$\DDD_w^\Delta$, i.e., $\DDD_w$ augmented with the transitions in $\Delta$, 
remains equivalent to $\DDD$.
We give a brief overview of how this is achieved.

Let $(u, a, v)$ denote a transition (meaning that state $v$ is reached from state $u$ by reading label $a$). If $\Delta \neq \emptyset$, we extract the transition $(u, a, v)$ at the top of $\Delta$. If adding $(u, a, v)$ to $\DDD_w$ does not violate any Wheeler axiom (Definition \ref{def: wheeler dfa}), then we perform such an addition and proceed with the next transition from $\Delta$. 
Otherwise, let $\mathcal I_v$ denote the set of strings labeling paths entering in state $v$ in the automaton $\DDD_w$ and either originating in the source or being left-infinite (i.e., originating in a loop). 
The Wheeler order $<$ on $\DDD_w$ can be equivalently characterized as follows: for any two states $u\neq v$, it holds $u<v$ if and only if $\mathcal I_u \prec \mathcal I_v$, meaning that all strings in $\mathcal I_{u}$ are co-lexicographically smaller than those in $\mathcal I_v$.
The fact that introducing $(u, a, v)$ violates a Wheeler axiom, means that there exists a state $p\neq v$ such that, after including the new transition, either (1) the co-lexicographic range of $\mathcal I_p$ is contained in that of $\mathcal I_v$, (2) the two ranges intersect but none is contained in the other.

Both cases can be resolved by (possibly) copying states of $\DDD_w$ and changing the destinations of some transitions; here we comment only on case (1) as it is simpler, but a similar technique applies to case (2). 
See also Figure \ref{fig:small example} for a simple example only involving case (1).
The axiom violation of case (1) can be solved by splitting the set $\mathcal I_v$ into two disjoint sets $\mathcal I_{v'}$ and $\mathcal I_{v''}$ such that $\mathcal I_{v'} \prec \mathcal I_p \prec \mathcal I_{v''}$. 
On $\DDD_w$, this can be achieved by creating a new state $v''$. The incoming transitions of $v$ stay unchanged, while the new transition $(u, a, v)$ is renamed into $(u, a, v'')$ and inserted in $\DDD_w$. 
This way, either $v < p < v''$ or $v'' < p < v$ hold in Wheeler order. 
Finally, for each outgoing transition $(v,c,q)$ of $v$, we insert in $\Delta$ a new transition $(v'',c,q)$. 
As a result, $v''$ becomes Myhill-Nerode equivalent to $v$ in the full automaton $\DDD_w^\Delta$. On the other hand, an algorithm invariant makes sure that 
no two adjacent states in Wheeler order are simultaneously Myhill-Nerode equivalent and have the same incoming label; 
as discussed below, this will be important for correctness and completeness.

The algorithm terminates when $\Delta$ becomes empty. Since each step may potentially add transitions to $\DDD_w^\Delta$, termination of the algorithm is not trivial. As a matter of fact, we prove that the algorithm terminates if and only if $\LLL(\DDD)$ is Wheeler; in that case, the algorithm's output $\DDD_w$ is the minimum Wheeler DFA recognizing $\LLL(\DDD)$. It is not hard to explain intuitively why this holds true. We leverage on the fact that, as proved by Alanko et al. in \cite{alanko2020regular}, a trimmed \footnote{Meaning that every state is reachable from the source and can reach at least one final state.} DFA is the minimum WDFA for its language if and only if adjacent states in Wheeler order sharing the same incoming label (unique by the Wheeler axioms), are not Myhill-Nerode equivalent. 
Since\footnote{As a matter of fact, these properties are a simplified description of the algorithm's invariants.} (i)~we start from the minimum DFA for the language, (ii)~our algorithm never inserts state $v$ near state $p$ (in Wheeler order) with $v$ being Myhill-Nerode equivalent to $p$ and sharing its incoming letter, 
(iii)~the Wheeler axioms hold true at every step on the transitions of $\DDD_w$, (iv)~every state is reachable from the source in $\DDD_w$, (v)~every state can reach a final state in $\DDD_w^\Delta$, and (vi)~$\DDD_w^\Delta$ recognizes $\LLL(\DDD)$ at every step, 
Alanko et al.'s characterization~\cite{alanko2020regular} allows us to conclude that, if the algorithm terminates, then $\DDD_w$ is indeed the minimum Wheeler DFA recognizing $\LLL(\DDD)$. 
The claimed complexity follows from the fact that $\DDD_w$ is stored using a dynamic data structure supporting updates in logarithmic time.

To prove that, if $\LLL(\DDD)$ is Wheeler, then the algorithm terminates, consider the minimum WDFA $\MMM$ of $\LLL(\DDD)$, and let $n_{\MMM}$ be the number of its states. 
If, for a contradiction, the algorithm does not terminate, then $\Delta$ is replenished infinitely often. In turn, this implies that the algorithm creates new states of $\DDD_w$ infinitely often. Upon creating the $(n_{\MMM}+1)$-th state, we pick a string $\alpha_i$ labeling a path from the source state to the $i$-th state (in Wheeler order) of $\DDD_w$, for all $i\in \{0,\dots, n_{\MMM}\}$ (this is possible by Property (iv) above). By Properties (v-vi), strings $\alpha_i$ must also belong to the prefix closure of $\LLL(\DDD)$, thereby they must label paths from the source to states $w_0, \dots, w_{n_\MMM}$ of $\MMM$. We finally argue that those states must be distinct, otherwise Property (ii) would imply that $\LLL(\MMM) \neq \LLL(\DDD)$. This means that $\MMM$ has $n_{\MMM}+1$ states, a contradiction.

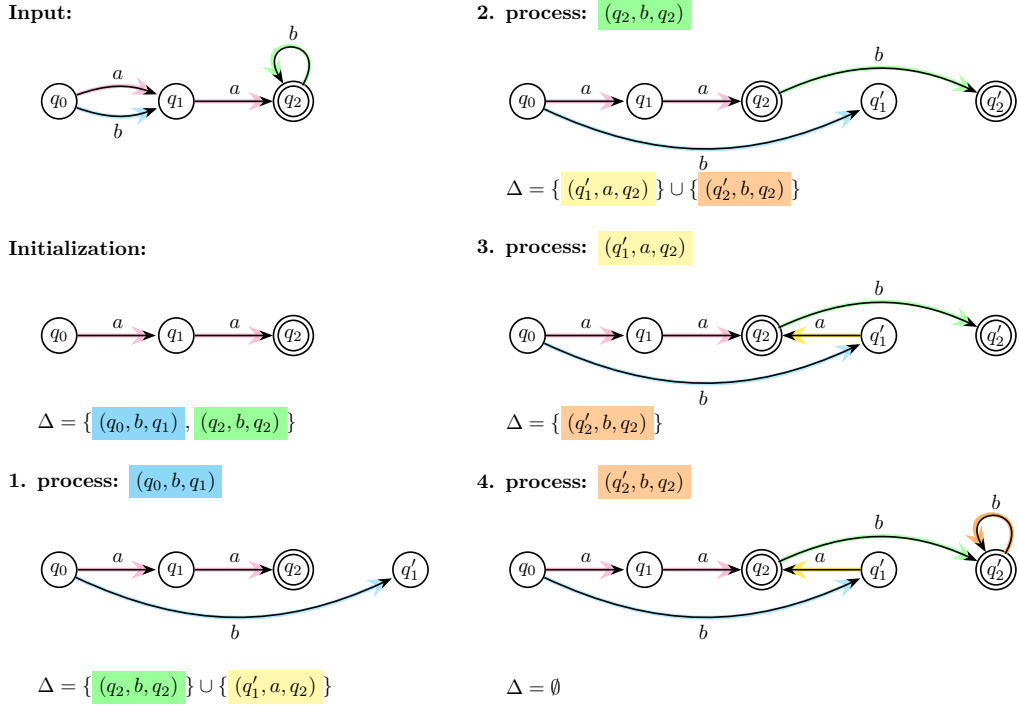
\begin{figure}[ht!]
    \centering
    \usetikzlibrary{automata, positioning, arrows.meta}

\begin{tikzpicture}[
    shorten >=1pt, node distance=2cm, on grid, auto,
    state/.style={circle, draw, minimum size=0.6cm, inner sep=0pt},
    accepting/.style={double, double distance=1.5pt},
    thick,
    >=Stealth,
    purple_arc/.style={draw=magenta!60, line width=2pt, opacity=0.5}, ,
    blue_arc/.style={draw=cyan!50, line width=2pt, opacity=0.5},
    green_arc/.style={draw=green!60, line width=2pt, opacity=0.5},
    yellow_arc/.style={draw=yellow!80!orange, line width=2pt, opacity=0.6},
    orange_arc/.style={draw=orange!70, line width=2pt, opacity=0.6},
    my_loop/.style={loop, looseness=8, out=60, in=120}
]

\node[anchor=west] at (-1, 1.5) {\textbf{Input:}};
\node[state] (q0) at (0, 0) {$q_0$};
\node[state, node distance=2cm] (q1) [right=of q0] {$q_1$};
\node[state, node distance=2cm, accepting] (q2) [right=of q1] {$q_2$};

\draw[purple_arc] (q0) edge[bend left=20] (q1);
\draw (q0) edge[bend left=20] node[above] {$a$} (q1);

\draw[blue_arc] (q0) edge[bend right=20] (q1);
\draw (q0) edge[bend right=20] node[below] {$b$} (q1);

\draw[purple_arc] (q1) -- (q2);
\draw (q1) -- node[above] {$a$} (q2);

\draw[green_arc] (q2) edge[my_loop] (q2);
\draw (q2) edge[my_loop] node[above] {$b$} (q2);

\begin{scope}[yshift=-4cm] 
    \node[anchor=west] at (-1, 1.5) {\textbf{Initialization:}};
    \node[state] (q0_i) at (0, 0) {$q_0$};
    \node[state] (q1_i) [right=of q0_i] {$q_1$};
    \node[state, accepting] (q2_i) [right=of q1_i] {$q_2$};
    
    \draw[purple_arc] (q0_i) -- (q1_i); \draw (q0_i) -- node[above] {$a$} (q1_i);
    \draw[purple_arc] (q1_i) -- (q2_i); \draw (q1_i) -- node[above] {$a$} (q2_i);
    
    \node[anchor=west] at (-0.5, -1.5) {$\Delta =\{\mathbox{cyan!40}{(q_0, b, q_1)}, \mathbox{green!40}{(q_2,b,q_2)} \}$};
\end{scope}

\begin{scope}[yshift=-8cm]
    \node[anchor=west] at (-1, 1.5) {\textbf{1. process: $\mathbox{cyan!40}{(q_0, b, q_1)}$}};
    \node[state] (q0_1) at (0, 0) {$q_0$};
    \node[state] (q1_1) [right=of q0_1] {$q_1$};
    \node[state, accepting] (q2_1) [right=of q1_1] {$q_2$};
    \node[state] (q1_prime_1) [right=2cm of q2_1] {$q_1'$};
    
    \draw[purple_arc] (q0_1) -- (q1_1); \draw (q0_1) -- node[above] {$a$} (q1_1);
    \draw[purple_arc] (q1_1) -- (q2_1); \draw (q1_1) -- node[above] {$a$} (q2_1);
    \draw[blue_arc] (q0_1) edge[bend right=25] (q1_prime_1);
    \draw (q0_1) edge[bend right=25] node[below] {$b$} (q1_prime_1);
    
    \node[anchor=west] at (-0.5, -2) {$\Delta = \{ \mathbox{green!40}{(q_2, b, q_2)} \} \cup \{ \mathbox{yellow!40}{(q_1', a, q_2)} \}$};
\end{scope}

\begin{scope}[xshift=8cm, yshift=0cm]
    \node[anchor=west] at (-1, 1.5) {\textbf{2. process: $\mathbox{green!40}{(q_2, b, q_2)}$}};
    \node[state] (q0_2) at (0, 0) {$q_0$};
    \node[state] (q1_2) [right=2cm of q0_2] {$q_1$};
    \node[state, accepting] (q2_2) [right=2cm of q1_2] {$q_2$};
    \node[state] (q1_p2) [right=2cm of q2_2] {$q_1'$};
    \node[state, accepting] (q2_p2) [right=2cm of q1_p2] {$q_2'$};
    
    \draw[purple_arc] (q0_2) -- (q1_2); \draw (q0_2) -- node[above] {$a$} (q1_2);
    \draw[purple_arc] (q1_2) -- (q2_2); \draw (q1_2) -- node[above] {$a$} (q2_2);
    \draw[blue_arc] (q0_2) edge[bend right=25] (q1_p2); \draw (q0_2) edge[bend right=25] node[below] {$b$} (q1_p2);
    \draw[green_arc] (q2_2) edge[bend left=25] (q2_p2); \draw (q2_2) edge[bend left=25] node[above] {$b$} (q2_p2);
    
    \node[anchor=west] at (-0.5, -1.5) {$\Delta = \{ \mathbox{yellow!40}{(q_1', a, q_2)} \} \cup \{ \mathbox{orange!40}{(q_2', b, q_2)} \}$};
\end{scope}

\begin{scope}[xshift=8cm, yshift=-4cm]
    \node[anchor=west] at (-1, 1.5) {\textbf{3. process: $\mathbox{yellow!40}{(q_1', a, q_2)}$}};
    \node[state] (q0_3) at (0, 0) {$q_0$};
    \node[state] (q1_3) [right=2cm of q0_3] {$q_1$};
    \node[state, accepting] (q2_3) [right=2cm of q1_3] {$q_2$};
    \node[state] (q1_p3) [right=2cm of q2_3] {$q_1'$};
    \node[state, accepting] (q2_p3) [right=2cm of q1_p3] {$q_2'$};
    
    \draw[purple_arc] (q0_3) -- (q1_3); \draw (q0_3) -- node[above] {$a$} (q1_3);
    \draw[purple_arc] (q1_3) -- (q2_3); \draw (q1_3) -- node[above] {$a$} (q2_3);
    \draw[blue_arc] (q0_3)  edge[bend right=25] (q1_p3); \draw (q0_3) edge[bend right=25] node[below] {$b$} (q1_p3);
    \draw[green_arc] (q2_3)  edge[bend left=25] (q2_p3); \draw (q2_3) edge[bend left=25] node[above] {$b$} (q2_p3);
    \draw[yellow_arc] (q1_p3) -- (q2_3); \draw (q1_p3) -- node[above] {$a$} (q2_3);
    
    \node[anchor=west] at (-0.5, -1.5) {$\Delta = \{ \mathbox{orange!40}{(q_2', b, q_2)} \}$};
\end{scope}

\begin{scope}[xshift=8cm, yshift=-8cm]
    \node[anchor=west] at (-1, 1.5) {\textbf{4. process: $\mathbox{orange!40}{(q_2', b, q_2)}$}};
    \node[state] (q0_4) at (0, 0) {$q_0$};
    \node[state] (q1_4) [right=2cm of q0_4] {$q_1$};
    \node[state, accepting] (q2_4) [right=2cm of q1_4] {$q_2$};
    \node[state] (q1_p4) [right=2cm of q2_4] {$q_1'$};
    \node[state, accepting] (q2_p4) [right=2cm of q1_p4] {$q_2'$};
    
    \draw[purple_arc] (q0_4) -- (q1_4); \draw (q0_4) -- node[above] {$a$} (q1_4);
    \draw[purple_arc] (q1_4) -- (q2_4);\draw (q1_4) -- node[above] {$a$} (q2_4);
    \draw[blue_arc] (q0_4)  edge[bend right=25] (q1_p4);\draw (q0_4) edge[bend right=25] node[below] {$b$} (q1_p4);
    \draw[green_arc] (q2_4)  edge[bend left=25] (q2_p4);\draw (q2_4) edge[bend left=25] node[above] {$b$} (q2_p4);
     \draw[yellow_arc] (q1_p4) -- (q2_4); \draw (q1_p4) -- node[above] {$a$} (q2_4);
    
    \draw[orange_arc] (q2_p4) edge[my_loop] (q2_p4);
    \draw[orange_arc] (q2_p4) edge[my_loop] (q2_p4);
    \draw (q2_p4) edge[my_loop] node[above] {$b$} (q2_p4);
    
    \node[anchor=west] at (-0.5, -2) {$\Delta = \emptyset$};
\end{scope}

\end{tikzpicture}
    \caption{
    The (non-Wheeler) DFA $\DDD$ (top left) is minimum for the Wheeler language $\LLL=(a|b)ab^*$. 
    \textbf{Initialization.}
    We start by choosing a spanning tree $\DDD_w$ rooted at the source, here containing the two pink transitions $(q_0, a, q_1)$ and $(q_1, a, q_2)$. The remaining two transitions $(q_0, b, q_1)$ and $(q_2, b, q_2)$ are pushed in the queue $\Delta$. At any step, we denote with $\DDD_w^\Delta$ the union of $\DDD_w$ and $\Delta$ (at the beginning, $\DDD_w^\Delta = \DDD$).
    \textbf{First step.} The algorithm pops the transition $(q_0, b, q_1)$ from $\Delta$ and processes it. The introduction of $(q_0, b, q_1)$ in $\DDD_w$ would add the string $b$ to the set of strings $\mathcal I_{q_1}$ entering in state $q_1$ in the sub-automaton $\DDD_w$ and the string $ba$ to $\mathcal I_{q_2}$; as a result, the co-lexicographic ranges of $\mathcal I_{q_1}=\{a, b\}$ and $\mathcal I_{q_2}=\{aa,ba\}$ would intersect ($a \prec aa \prec ba \prec b$).
    Hence, the transition is instead introduced pointing to a new state $q_1'$ that acts as a copy of state $q_1$. We now have $\mathcal I_{q_1}=\{a\}$ and $\mathcal I_{q_1'}=\{b\}$ and hence both these sets do not intersect co-lexicographically with $\mathcal I_{q_2}=\{aa\}$. As $q_1$ and $q_1'$ need to be Myhill-Nerode equivalent in the full automaton $\DDD_w^\Delta$, the algorithm needs to add corresponding transitions $(q_1', c, q)$ for all outgoing transitions $(q_1, c, q)$ of $q_1$ in $\DDD_w^\Delta$ to the set of transitions $\Delta$ to be processed. In this particular case, the algorithm adds $(q_1', a, q_2)$ to $\Delta$. 
    \textbf{Second step.}
    Next, the algorithm pops the transition $(q_2, b, q_2)$ from $\Delta$ and processes it. This leads to the creation of a new state $q_2'$; this is required since otherwise $\mathcal I_{q_1'}$ and $\mathcal I_{q_2}$ would intersect co-lexicographically. The transition $(q_2', b, q_2)$ is then added to $\Delta$ in order to make $q_2$ and $q_2'$ Myhill-Nerode equivalent. 
    \textbf{Third and fourth steps.}
    Finally, the transitions $(q_1', a,q_2)$ and $(q_2', b, q_2)$ from $\Delta$ can be simply added to $\DDD_w$ without breaking the Wheeler order, hence $\Delta$ becomes empty and the algorithm outputs the minimum Wheeler DFA $\DDD_w$ (bottom right) accepting $\LLL$.}
    \label{fig:small example}
\end{figure}

\section{Preliminaries}

We denote by $[N]$ the set of integers $\{1,2,\dots,N\}$. 
We recall that a \emph{total order} is a binary relation over a set $U$ that is reflexive, transitive, antisymmetric, and strongly connected. For a total order $\preceq$ on a set $U$, if $x\preceq y$ and $x\neq y$ for two elements $x,y\in U$, we write $x\prec y$.
Let $\Sigma$ be a finite alphabet of size $\sigma$ equipped with a total order $\preceq$. Without loss of generality we assume $\Sigma = [\sigma]$ to be the first $\sigma$ integers with the natural total order. 
Symbol $\Sigma^*$ denotes the set of all finite strings (including the empty string $\epsilon$) that can be formed by concatenating symbols from $\Sigma$.
With a slight abuse of notation, we use $\preceq$ also to denote the \emph{co-lexicographic (or co-lex) order} between strings that is defined (recursively) as follows: For two strings $\alpha,\beta \in \Sigma^*$, we have $\alpha \preceq \beta$ if
(i)~$\alpha = \epsilon$ or if
(ii)~$\alpha = \alpha^\prime a$ and $\beta = \beta^\prime b $ with $a,b \in \Sigma$ and $\alpha',\beta' \in \Sigma^*$ such that $a \prec b$ or $a = b \land \alpha^\prime \preceq \beta^\prime$.

\begin{definition}[DFA]
A \emph{Deterministic Finite Automaton (DFA)} $\mathcal A$ is a 5-tuple
$\AAA = (Q, \Sigma, \delta, q_0, F)$, where $Q$ is a finite set of \emph{states}, $\Sigma$ is a finite alphabet, $\delta : Q \times \Sigma \to Q$ is the \emph{transition function}, $q_0 \in Q$ is the \emph{initial state}, and $F \subseteq Q$ is the set of \emph{final states}. 
\end{definition}

With a slight abuse of notation, we sometimes consider $\delta$ as the set of triples $\{(u,a,v)\in Q\times \Sigma\times Q: v = \delta(u,a)\}$.
We extend the domain of the transition function to strings as usual: for $a \in \Sigma$, $\alpha \in \Sigma^*$, and $q \in Q$ we define $\delta(q,\epsilon)=q$ and $\delta(q, a \alpha) = \delta(\delta(q, a),\alpha)$.

For a state $q \in Q$, we denote with $I_\AAA(q)$ the set of strings \footnote{The set $I_\AAA(q)$ has not to be confused with the set $\mathcal I_q$ mentioned in the introduction, which is instead the union of $I_{\DDD_w}(q)$ and all left-infinite strings reaching state $q$ in $\DDD_w$} \textit{reaching} $q$ from the initial state in $\AAA$, i.e., $I_\AAA(q) := \{\alpha\in\Sigma^* : \delta(q_0, \alpha) = q\}$.
The \textit{language} $\mathcal{L(A)} \subseteq \Sigma^*$ \textit{recognized by} $\AAA$ is the set of strings reaching a final state from the initial state, that is $\mathcal{L(A)} = \bigcup_{q \in F}I_\AAA(q)$. With $\Pref(\LLL(\AAA))$ we denote the set of prefixes of words in $\LLL(\AAA)$. 

A DFA is \emph{accessible} if $I_q\neq \emptyset$ for all states, i.e., every state is reachable by a directed path from the source state. A DFA is instead \emph{co-accessible}, if every state can reach a final state. A DFA that is both accessible and co-accessible is called \emph{trimmed}.

Two strings $\alpha,\beta$ are \emph{Myhill-Nerode equivalent} (MN) in $\AAA$, denoted by $\alpha \equiv_\AAA \beta$, if for all $\gamma \in \Sigma^*$, it holds $\alpha\cdot \gamma \in \LLL(\AAA)$ if and only if $\beta\cdot \gamma \in \LLL(\AAA)$~\cite{myhill1957finite, nerode1958linear}. 
This equivalence relation naturally extends to states: two states $u$ and $v$ of $\AAA$ are Myhill-Nerode equivalent, denoted by $u \equiv_\AAA v$, if for all $\gamma \in \Sigma^*$, it holds that the state $\delta(u,\gamma)$ is final if and only if $\delta(v,\gamma)$ is. We omit the subscript from $\equiv_\AAA$ if it is clear from the context.
The equivalence classes with respect to the Myhill-Nerode equivalence relation (in both variants: strings/states) correspond exactly to the states of the unique minimum DFA that accepts the same language~\cite{myhill1957finite, nerode1958linear}.

\emph{Wheeler automata} have been introduced by 
Gagie et al.~\cite{GAGIE201767} as a tool extending the celebrated Burrows-Wheeler transform \cite{burrows1994block} from strings to labeled graphs. In this paper, we consider only deterministic Wheeler automata. In the following definition, we assume that the initial state $q_0$ is the only one with no incoming transitions (this will simplify our description and, as we show below, comes with no loss of generality). 

\begin{definition}[Wheeler DFA~\cite{GAGIE201767}]
\label{def: wheeler dfa}
A Wheeler DFA (WDFA) is a DFA $\mathcal{A} = (Q, \Sigma, \delta, q_0, F)$ such that $q_0$ is the unique state with no incoming transitions. In addition, the automaton must admit a (unique) total order $\le$ on $Q$, called the Wheeler order, where $q_0$ is the minimum element. This order is defined such that for any two transitions $(u, a, u')$ and $(v, b, v')$, the following Wheeler Properties (a.k.a. \emph{Wheeler Axioms}) hold:
    \begin{enumerate}[{W}1.]
        \item \label{item: wp different letter}
        If $a \prec b$, then $u' < v'$. 
        \item \label{item: wp same letter}
        If $a = b$ and $u < v$, then $u' \le v'$. 
    \end{enumerate}
\end{definition}

The fact that WDFA admit a \emph{unique} Wheeler order is known \cite{alanko2020regular} and significantly simplifies the problem considered in this paper (Wheeler NFA, on the other hand, may admit multiple Wheeler orders).
We notice that Property W\ref{item: wp different letter} implies that a WDFA $\AAA$ is always \emph{input-consistent}, meaning that for all $u,v \in Q$ and $a, b\in \Sigma$, if $\delta(u,a) = \delta(v,b)$, then $a = b$.
On WDFA, we denote by $\lambda_\AAA : Q \to \Sigma$ the function that returns the unique incoming label of a state, that is, for a state $v\in Q \setminus \{q_0\}$ we define $ \lambda_\AAA(v) = a$ if and only if for all $u \in Q$ with $\delta(u,b) = v$ it holds that $b = a$.
For the initial state, we define $\lambda_\AAA(q_0) = \# \not\in \Sigma$, where $\#$ is an artificial letter strictly smaller than all characters in $\Sigma$ according to $\preceq$. 
When $\AAA$ is clear from the context, we drop the subscript and simply write $\lambda$.
The assumption that $q_0$ has no incoming transitions is without loss of generality as one can always modify any regular language $\LLL$ into $\$\LLL$ for a new character $\$$ not belonging to the original alphabet of $\LLL$. As a result, any DFA $\mathcal A$ recognizing $\LLL$ can be converted into one recognizing $\$\LLL$ by just adding to it a new initial state $q_0'$ without incoming transitions and a transition $(q_0',\$,q_0)$. It is not hard to show that this transformation also preserves Wheelerness, if $\mathcal A$ was Wheeler.

The following lemma states that the Wheeler order of the states is consistent with the co-lexicographic order of the strings reaching states.
\begin{lemma}[Lemma 3 in \cite{ConteDCC2023}]
\label{lem: wheeler order}
    Let $\AAA=(Q, \Sigma, \delta, q_0, F)$ be a WDFA with Wheeler order $\le$ and let $u,v\in Q$ and $\alpha\in I_\AAA(u)$ as well as $\beta\in I_\AAA(v)$. Then $u<v$ implies $\alpha \prec \beta$.
\end{lemma}

Alanko et al.~\cite{alanko2020regular} prove the following \textit{Myhill-Nerode theorem} for Wheeler languages.
\begin{theorem}[Minimum Wheeler DFA, Theorem 4.2~in~\cite{alanko2020regular}] 
\label{thm: wheeler minimality}
    Let $\DDD_w$ be a trimmed Wheeler DFA with Wheeler order $\le$ and states $q_0 < q_1 < \dots < q_t$. Furthermore, let $\equiv$ be the Myhill-Nerode equivalence relation among its states. Then, $\DDD_w$ is the minimum Wheeler DFA recognizing $\mathcal{L(D}_w\mathcal{)}$ if and only if $q_{i - 1} \equiv q_{i}$ implies $\lambda(q_{i-1}) \neq \lambda(q_{i})$ for every $i\in [t]$.
\end{theorem}
Note that Theorem 4.2 in the paper Alanko et al.~\cite{alanko2020regular} does not explicitly mention the assumption that $\DDD_w$ needs to be trimmed, instead this necessary assumption is globally assumed (in Section 1.2) for all automata in their paper.

\section{Formal Description of the Algorithm}
\label{sec:algorithm}
\cref{algorithm: min wheeler} is our solution to \cref{problem:Wheelerization}. 
In this section, we describe all the details of the algorithm, while in the next section, we prove its correctness, completeness, and runtime. 
For convenience, Table~\ref{table:cheatsheet} summarizes the meaning of all symbols used in the algorithm and its analysis.
\begin{table}[ht]
    \centering
    \begin{tabular}{|c|p{8.8cm}|}
    \hline
        $\DDD = (Q=[n],\Sigma, \delta, q_0, F)$ & Input DFA. If $\delta(u,a)$ is not defined, we write $\delta(u,a)=\bot$. We let $\delta(\bot,a)=\bot$.\\\hline
        $\DDD_w = (Q_w, \Sigma, \delta_w, q_0, F_w)$ & Intermediate and output WDFA. At the beginning, $Q_w=Q$, $F_w=F$, and $\delta_w \subseteq \delta$ is a spanning tree of $\DDD$ rooted in $q_0$. Invariants: $\DDD_w$ is Wheeler and accessible. \\\hline
        $\lambda(u)$ & Shorthand for $\lambda_{\DDD_w}(u)$: label of incoming transitions of $u\in Q_w$.\\\hline
        $\Delta$ & Queue with transitions to be added to $\DDD_w$. \\\hline
        $\DDD_w^\Delta = (Q_w, \Sigma, \delta_w \cup \Delta, q_0, F_w)$ & $\DDD_w$ augmented with $\Delta$. Initially, $\DDD_w^\Delta = \DDD$ is the input. At the end, $\DDD_w^\Delta = \DDD_w$ is the algorithm's output. Invariants: $\LLL(\DDD_w^\Delta) = \LLL(\DDD)$ and $\DDD_w^\Delta$ is deterministic and  co-accessible. \\\hline
       $\equiv$  & Myhill-Nerode equivalence relation on $\DDD_w^{\Delta}$. \\\hline
       $\cong$  & Equivalence relation on $\DDD_w^{\Delta}$ computed by the algorithm. Invariant: $\cong$ and $\equiv$ are equal. \\\hline
       $\LEX$ & Wheeler order (permutation) of $Q_w = \{1, \dots, |Q_w|\}$. \\\hline
       $\LEX^{-1}$ & Inverse of $\LEX$. We define $\LEX^{-1}[\bot]=0$.\\\hline
      $\texttt{pred}(u,a)$ & Largest state $v<u$ in Wheeler ($\LEX$) order such that $\delta_w(v,a)\neq \bot$; $\texttt{pred}(u,a)=\bot$ if no such state exists.\\\hline
       $\texttt{succ}(u,a)$ & minimum state $v>u$ in Wheeler ($\LEX$) order such that $\delta_w(v,a)\neq \bot$; $\texttt{succ}(u,a)=\bot$ if no such state exists.\\\hline
       \texttt{insert($u,i$)} & Inserts $u$ at position $i$ in $\LEX$, i.e. \newline $\LEX \coloneqq \LEX[\dots, i-1]\ u\ \LEX[i,\dots]$.\\\hline
       \texttt{copy($q,T,j$)} & Creates a $\cong$-equivalent copy $q'$ of $q$ and inserts it in position $\LEX[j]$. Outgoing transitions of $q$ belonging to $\Delta$ (resp. $\delta_w$) are copied on $q'$ and inserted in $\Delta$ (resp. $T$).\\\hline
    \end{tabular}
    \vspace{5pt}
    \caption{
    Cheatsheet of symbols and functions used in \cref{algorithm: min wheeler}.}
    \label{table:cheatsheet}
\end{table}

Let us recall that, given a minimum DFA $\DDD=(Q, \Sigma,\delta, q_0, F)$ for a Wheeler language $\LLL$, the goal is to construct the minimum Wheeler DFA $\DDD_w$ accepting $\LLL$. 
We assume that states are represented by integers: $Q = [n]$, and that the special symbol $\bot$ (non-existing state) corresponds to integer $\bot = 0$.
\cref{algorithm: min wheeler} is composed of two functions: function \texttt{MinWheeler} (the main algorithm) and the auxiliary function \texttt{copy}. We start from the former.

\paragraph*{Initialization.}
Function \texttt{MinWheeler} begins in Line \ref{line:SortedSpanningTree} by calling \texttt{SortedSpanningTree($\DDD$,$q_0$)}. This function (1) computes a spanning tree of $\DDD =(Q, \Sigma, \delta, q_0, F)$ rooted in $q_0$ (e.g.\ via a DFS visit) and (2) sorts $Q$ according to the colexicographic order of the strings labeling the source-to-nodes paths in the spanning tree. Here, the linear-time algorithm~\cite[Theorem 2]{ferragina2009compressing} can be used. The result is returned as a pair $(\delta_w,\LEX)$, where $\delta_w\subseteq \delta$ is the set of transitions forming the spanning tree and $\LEX[1,n]$ is the permutation of $Q=[n]$ containing the states $Q$ sorted by the colexicographic order of the source-to-state strings in the spanning tree. In other words, $\LEX$ encodes the Wheeler order of the Wheeler subgraph $\DDD_w = (Q,\Sigma,\delta_w,q_0,F_w)$, where $F_w$ is initially set to be equal to $F$ (Line \ref{line:set Fw}). At any point of the execution, $\LEX$ will be a permutation, and we will indicate with $\LEX^{-1}[q]$ the position $i$ such that $\LEX[i]=q$. As a special case, we define $\LEX^{-1}[\bot]=0$, where $\bot$ indicates a non-existing state. 

In Line \ref{line:init delta}, the algorithm initializes the queue $\Delta$ to contain all transitions in $\delta\setminus \delta_w$. At any step of the algorithm's execution, we denote with $\DDD_w^\Delta = (Q,\Sigma,\delta_w \cup \Delta,q_0,F_w)$ the union between $\DDD_w$ and the transitions in $\Delta$. Note that, at this point of the execution, $\DDD_w^\Delta = \DDD$.

The algorithm stores an internal representation of the Myhill-Nerode equivalence relation between states of $\DDD_w^\Delta$ as an integer label indicated with $[u]_{\cong}$, associated to each state $u$; Line~\ref{line: init congr MN} initializes $[u]_{\cong} := u$ to be equal to the original Myhill-Nerode equivalence relation $\equiv$ on $\DDD$ (recall that $\DDD$ is minimum, so each state is a class of  $\equiv$); note that $[\bot]_{\cong} := \bot \notin Q$.
Observe that in our pseudocode we use a symbol ($\cong$) different from $\equiv$ in order to distinguish the relation computed by the algorithm with the ``real'' Myhill-Nerode relation; ultimately, we will prove that $\cong$ and $\equiv$ are the same relation on $\DDD_w^{\Delta}$ before and after every iteration of the \texttt{while} loop, but this distinction will be useful in our proofs. We emphasize that $\equiv$ and $\cong$ are always to be read as equivalence relations on $\DDD_w^{\Delta}$ (and not $\DDD_w$).

\paragraph*{\texttt{while} loop.}

\subparagraph*{Identifying the predecessor and successor of $u$ by letter $a$.}
The \texttt{while} loop begins by popping a transition $(u,a,v)$ from $\Delta$
(Line~\ref{line:pop}). Then, Line~\ref{line:pred} (resp.~\ref{line:succ})
identifies the largest (resp.~smallest) state $\texttt{pred}(u,a)<u$
(resp.~$\texttt{succ}(u,a)>u$) in Wheeler order (i.e., in $\LEX$) having an
outgoing transition labeled with letter $a$ in $\DDD_w$, that is, such that
$p=\delta_w(\texttt{pred}(u,a),a)$ (resp.~$s=\delta_w(\texttt{succ}(u,a),a)$)
is defined. If no such state exists, then $p=\bot$ (resp.~$s=\bot$).

We claim that, by determinism and the Wheeler axioms, if both $p$ and $s$
exist then no state lies strictly between them in Wheeler order; in
particular, either $p=s$ or $p$ and $s$ are adjacent in $\LEX$. Indeed, since
$\DDD_w$ is deterministic, the only $a$-labeled transition leaving $u$ is
$(u,a,v)$, which at this point belongs to $\Delta$ rather than to $\delta_w$;
hence $u$ has \emph{no} outgoing $a$-transition in $\DDD_w$. By definition of
$\texttt{pred}(u,a)$ and $\texttt{succ}(u,a)$, the only state that could lie
strictly between them in Wheeler order and carry an outgoing $a$-transition is
$u$ itself, which carries none. Thus $\texttt{pred}(u,a)$ and
$\texttt{succ}(u,a)$ are \emph{consecutive} among the states with an outgoing
$a$-transition in $\DDD_w$. Now, Wheeler axiom~W1 forces the targets of
$a$-labeled transitions to form a contiguous range of $\LEX$, while W2 makes
the map sending each such source to its $a$-target monotone (non-decreasing)
and onto this range. Two consecutive sources are therefore mapped to two
targets with no state of the range strictly between them, which is precisely
the claim for $p=\delta_w(\texttt{pred}(u,a),a)$ and
$s=\delta_w(\texttt{succ}(u,a),a)$.

In order to satisfy the Wheeler axioms, the edge $(u,v)$ should point to a
state immediately after $p$ (if it exists) and before $s$ (if it exists).

Next, we describe the function \texttt{copy($q,T,j$)} since understanding its behavior is important for the subsequent steps of the algorithm.  

\subparagraph*{Function \texttt{copy($q,T,j$)}.} In Line \ref{line:q'}, the function creates a copy $q'$ of $q$ initialized to be the next available integer $|\LEX|+1$ (at any point, states are consecutive integers: $\{1, \dots, |\LEX|\}$), inserts it in position $j$ of $\LEX$ (Line \ref{line:LEX insert}), 
copies the ``final'' status of $q$ and its Myhill-Nerode equivalence into $q'$ (Lines \ref{line:set final} and \ref{line:set equiv}),
and copies the outgoing transitions of $q$ into $q'$ (Lines \ref{line:modify Delta} and \ref{line:modify T}). 
More in detail, Line \ref{line:modify Delta} copies each outgoing transition $(q,c,x)$ of $q$ belonging to $\Delta$, inserting the copy $(q',c,x)$ in $\Delta$. Line \ref{line:modify T}, on the other hand, takes care of the outgoing transitions of $q$ belonging to $\delta_w$. For each such transition $(q,c,x)\in \delta_w$, the new transition $(q',c,x)$ is inserted in the set $T$ passed \emph{by reference} to the function; this detail is important since the function will be called once by binding $T$ to $\delta_w$, and once by binding it to $\Delta$. 
We can now go back to describing the three cases that can occur in Lines \ref{line:if1}-\ref{line:CopyState2} of function \texttt{MinWheeler}.

\subparagraph*{Case 1: $v \cong p$ or $v \cong s$.} In other words, the \texttt{if} condition at line \ref{line:if1} succeeds. This case is the simplest: Lines \ref{line:choose v'} and \ref{line:insert transition} insert the transition $(u,a,v')$ into $\delta_w$, where $v'$ is the state among $p$ and $s$ being $\cong$-equivalent to $v$ (it could be $v = v'$). This case does not require creating new states, since adding the transition $(u,a,v')$ (equivalent to $(u,a,v)$ from a language perspective) to $\DDD_w$ does not violate any Wheeler axiom (we will prove this later).

\subparagraph*{Case 2: $p \not\cong v \not\cong s$ and $p = s \neq \bot$.} In other words, the \texttt{if} condition at line \ref{line:if1} does not succeed and the \texttt{if} condition at line \ref{line:if2} succeeds. 
We cannot simply move $(u,a,v)$ from $\Delta$ to $\delta_w$:
Wheeler axiom W2 would require $p<v<s$, which in this case is impossible since $v \neq p = s$. 
This issue can be solved by first splitting $p=s$ into two distinct $\cong$-equivalent states $p\neq s$. This is precisely the role of Line \ref{line:CopyState1}, which overwrites variable $s$ with a brand new state with the same outgoing transitions as $p$ by calling function \texttt{copy}. Importantly, this call to \texttt{copy} 
treats differently the outgoing transitions of $p$ belonging to $\Delta$ and those belonging to $\delta_w$. For each $(p,c,x)\in \Delta$ of the former type, a new transition $(s,c,x)$ is inserted into $\Delta$. For each $(p,c,x)\in \delta_w$ of the latter type, a new transition $(s,c,x)$ is inserted into $\delta_w$. This behavior is achieved by binding the formal parameter $T$ of function \texttt{copy} to the actual parameter $\delta_w$.
As we will prove later, this update of $\delta_w$ does not violate any Wheeler axiom and is crucial for the completeness of our algorithm. At this point, Line \ref{line:foreach2} replaces every transition $(x,a,p)\in \delta_w$ such that $u<x$ (in Wheeler order), with transition $(x,a,s)$. In other words, we ``move'' the destinations of all in-transitions of $p$ coming from a state $x$ larger than $u$ to $s$. As we will prove later, this preserves the Wheeler axioms and the language. 

We now have two adjacent (in Wheeler order) states $p < s$ being Myhill-Nerode equivalent, i.e. $p\cong s$. We are left to insert a new state $v'\cong v$ between them and add transition $(u,a,v')$ to $\delta_w$. This is precisely what Lines \ref{line:find i}, \ref{line:CopyState2}, and \ref{line:insert transition} do. 

First, note that symbol $\lambda$ in Line \ref{line:find i} indicates $\lambda_{\DDD_w}$; the subscript omission does not create ambiguity, since this notation is defined only on WDFA, and $\DDD_w$ is the only WDFA here. 
Line \ref{line:find i}, in this particular case ($p$ and $s$ both exist), simplifies to $i\coloneqq 1 + \LEX^{-1}[p]\ (=\LEX^{-1}[s])$ and therefore computes the position of $s$ in $\LEX$. 
To see this, observe that $\{j\ :\ \lambda(\LEX[j])\prec a\}$ is the set of indices in $\LEX$ of all states with incoming transitions (in $\delta_w$) labeled by characters strictly smaller than $a$, hence by Wheeler axiom W1 they precede $p$ (reached by $a$) in $\LEX$; as a result, the $\max$ operator returns $\LEX^{-1}[p]$. Line \ref{line:CopyState2}, with a call to \texttt{copy}, creates a new state $v'$, places it between $p$ and $s$ (so that $p<v'<s$ are adjacent), and copies the outgoing transitions of $v$ into $v'$, inserting them into $\Delta$ (in this case, we cannot add transitions to $\delta_w$ as they could violate some Wheeler axiom; inserting them in $\Delta$ ensures that they will be processed later). 
This is achieved by binding the formal parameter $T$ of \texttt{copy} to the actual parameter $\Delta$.
Finally, Line \ref{line:insert transition} adds $(u,a,v')$ to $\delta_w$.

\subparagraph*{Case 3: $p \not\cong v \not\cong s$ and $(p \neq s \vee s=\bot)$.} 
In other words, both the \texttt{if} conditions at lines \ref{line:if1} and \ref{line:if2} do not succeed.
Since $p \not\cong v \not\cong s$, we must place a $\cong$-equivalent copy $v'$ of $v$ immediately after $p$, if such a state exists (and before $s$, if it exists; note that, if both $p$ and $s$ exist, they must be adjacent in Wheeler order). If $p$ does not exist, then $v'$ has to become the first state in Wheeler order such that $\lambda(v')=a$. In both cases ($p \neq \bot$ or $p = \bot$), Line \ref{line:find i} finds the position $i$ where $v'$ has to be inserted: if $p\neq \bot$ then $\LEX^{-1}[p]$ is larger than all indices in $\{j\ :\ \lambda(\LEX[j])\prec a\}$, thereby Line \ref{line:find i} correctly computes $i \coloneqq 1+\LEX^{-1}[p]$; on the other hand, if $p=\bot$ then $\LEX^{-1}[p] = \LEX^{-1}[\bot] = 0$ and Line \ref{line:find i} computes $i\coloneqq 1+ \max(\{j\ :\ \lambda(\LEX[j])\prec a\} \cup \{0\})$. This means that $i-1$ is the position of the last state in $\LEX$ with incoming letter smaller than $a$ or, if all states in $\LEX$ have incoming letter larger than or equal to $a$, then $i-1 = 0$. Thus, position $i$ is where the first state with incoming letter $a$ should be inserted in $\LEX$ according to Wheeler Axiom~W1. 

After that, Line \ref{line:CopyState2}, with a call to \texttt{copy}, creates a new state $v'$, inserts it in position $i$ of $\LEX$, and copies the outgoing transitions of $v$ into $v'$, inserting them into $\Delta$. Finally, Line \ref{line:insert transition} adds transition $(u,a,v')$ to $\delta_w$.

\begin{algorithm}[!ht]
\LinesNumbered
\DontPrintSemicolon
\SetKwProg{Fn}{Function}{:}{}
\SetKwFunction{fcopy}{copy}
\SetKwFunction{fMinWheeler}{MinWheeler}

\caption{\minWheeler}\label{algorithm: min wheeler}
    \Input{minimum DFA $\DDD$ accepting Wheeler language $\LLL(\DDD)$}
    \Output{minimum Wheeler DFA $\DDD_w$ such that $\LLL(\DDD_w) = \LLL(\DDD)$}
    \medskip

\SetKwProg{Fn}{Function}{}{}
\Fn{\fcopy{$q$,$T$,$j$}: \tcp*[f]{\small $T$ (transitions set) passed by reference}}{
    $q' \coloneqq |\LEX|+1$\label{line:q'}\tcp*[r]{\small $\LEX$, $F_w$, $\Delta$, $\delta_w$: global variables}
    \texttt{insert($q',j$)}\label{line:LEX insert}\;
    \textbf{if} $q \in F_w$ \textbf{then} $F_w \coloneqq F_w \cup \{q'\}$\label{line:set final}\label{line:final q}\;
    $[q']_{\cong} \coloneqq [q]_{\cong}$\label{line:set equiv}\;
    $\Delta \coloneqq \Delta \cup \{(q',c,x)\ :\ (q,c,x) \in \Delta\}$\label{line:modify Delta}\;
    $T \coloneqq T \cup \{(q',c,x)\ :\ (q,c,x) \in \delta_w\}$\label{line:modify T}\;
    \Return{$q'$}\;
}

\Fn{\fMinWheeler{$\DDD =(Q, \Sigma, \delta, q_0, F)$}:}{

    $(\delta_w, \LEX) \coloneqq \texttt{SortedSpanningTree}(\DDD,q_0)$\label{line:SortedSpanningTree}\;

    $F_w \coloneqq F$\label{line:set Fw}\;

    $\Delta \coloneqq \delta \setminus \delta_w$\label{line:init delta}\;

    \textbf{foreach} $u\in Q\cup \{\bot\}$ \textbf{do} $[u]_{\cong} \coloneqq u$ \tcp*[r]{\small $[u]_{\cong}$ is an integer associated with $u$} \label{line: init congr MN}    
    \smallskip    
    
    \While{$\Delta \neq \emptyset$}{ \label{line:while}
        $(u,a,v) \coloneqq \Delta$.pop()\label{line:pop}\;
        $p \coloneqq \delta_w(\texttt{pred}(u,a),a)$\label{line:pred}\tcp*[r]{\small if $\texttt{pred}(u,a)=\bot$ then $p=\bot$ (same for $s$).}
        $s \coloneqq \delta_w(\texttt{succ}(u,a),a)$\label{line:succ}\;
        
        \smallskip

        \uIf{$[v]_{\cong} \in \{[p]_{\cong},[s]_{\cong}\}$\label{line:if1}}{
            \textbf{let} $v'\in\{p,s\}$ \textbf{be such that} $v \cong v'$\label{line:choose v'}\tcp*[r]{\small could be $v'=v$}   
        }\Else{

            \If{$p=s\neq \bot$\label{line:if2}}{
                $s \coloneqq \texttt{copy}(p,\delta_w,\LEX^{-1}[p]+1)$\label{line:CopyState1}\;

                \textbf{foreach} $(x,a,p)\in \delta_w$ s.t. $u < x$ \textbf{do} $\delta_w \coloneqq \left(\delta_w \setminus \{(x,a,p)\}\right) \cup \{(x,a,s)\}$\label{line:foreach2}\;
                
            }
            
            $i \coloneqq 1 + \max \big(\{j\ :\ \lambda(\LEX[j])\prec a\}\cup\{\LEX^{-1}[p]\}\big)$\label{line:find i}\tcp*[r]{$\LEX^{-1}[\bot]=0$}

            $v' \coloneqq \texttt{copy}(v,\Delta,i)$\label{line:CopyState2}\;
            
        }

        $\delta_w \coloneqq \delta_w \cup \{(u,a,v')\}$\label{line:insert transition}\;

    }
    \Return{$\DDD_w = (\{1,\dots,|\LEX|\}, \Sigma, \delta_w, q_0, F_w)$}\label{line:return}\;
}

\end{algorithm}

\section{Analysis: Correctness, Completeness, and Runtime}
\label{sec:analysis}

In this section we prove the correctness, completeness, and the bound on the runtime of Algorithm~\ref{algorithm: min wheeler}. In what follows, we fix an input DFA $\DDD = (Q, \Sigma, \delta, q_0, F)$ that is minimum for the Wheeler language $\LLL$ that it accepts and denote with $\DDD_w$ the output DFA of \minWheeler. Our goal is to prove the following theorem from the introduction, restated here.

\maintheorem*

The proof of \cref{thm:main theorem} relies on three main results. In \cref{lemma: termination} we show that the algorithm always terminates (assuming that $\LLL$ is a Wheeler language). By \cref{lemma: correctness refactored}, the output automaton $\DDD_w$ is deterministic, recognizes the same language $\LLL$, and is the minimum Wheeler DFA that accepts $\LLL$. Finally at the end of this, we show how to implement \minWheeler using dynamic data structures achieving the above claimed running time.

To facilitate the proofs of completeness and correctness, we first state the invariants maintained by the algorithm at the beginning of each iteration of the while loop. Due to space limitations, the proof of the lemma is deferred to Appendix~\ref{sec: invariant_proofs}.

\begin{restatable}[Invariants]{lemma}{invariantsrefactored}
\label{lemma: invariants refactored}
At any step of Algorithm~\ref{algorithm: min wheeler}, let us denote:
\begin{itemize}
    \item $\DDD_w = (Q_w=\LEX, \Sigma, \delta_w, q_0, F_w)$,
    \item $\DDD_w^\Delta = (Q_w=\LEX, \Sigma, \delta_w \cup \Delta, q_0, F_w)$, and
    \item $u \cong v$, 
    with $u,v \in Q_w$, if and only if $[u]_{\cong} = [v]_{\cong}$.
\end{itemize}
The following invariants hold before and after every iteration of the \texttt{while} loop of Algorithm~\ref{algorithm: min wheeler}:
\begin{enumerate}[(1)]
    \item \label{Loop Invariant: Deterministic}
    \textbf{Determinism:} $\DDD_w^\Delta$ is deterministic.
    \item \label{Loop Invariant: Connectivity}
    \textbf{Accessibility:} The automaton $\DDD_w$ is accessible.
    \item \label{Loop Invariant: wheeler order}
    \textbf{Wheeler Order:} $\DDD_w$ is Wheeler and $\LEX$ encodes its Wheeler order.
    \item \label{Loop Invariant: MN-equiv} \textbf{MN Equivalence:} $\equiv$ and $\cong$ are the same equivalence relation in the automaton $\DDD_w^\Delta$. 
    \item \label{Loop Invariant: language}
    \textbf{Language: } $\LLL(\DDD_w^\Delta) = \LLL(\DDD)$.
    \item \label{Loop Invariant: min order}
    \textbf{Wheeler-minimality:} 
    For any $i \in \{1,\dots, |\LEX|-1\}$, if $\LEX[i] \equiv \LEX[i+1]$ in $\DDD_w^\Delta$, then $\lambda_{\DDD_w}(\LEX[i]) \neq \lambda_{\DDD_w}(\LEX[i+1])$.
    \item \label{Loop Invariant: co-accessible}
    \textbf{Co-accessibility:} The automaton $\DDD_w^\Delta$ is co-accessible. 
    \end{enumerate}
\end{restatable}

\paragraph*{Completeness.}
Using the above invariants, we now prove that the algorithm always terminates for a Wheeler language $\LLL(\DDD)$ in input. 

\begin{lemma}[Completeness]
    \label{lemma: termination}
    \cref{algorithm: min wheeler} terminates for every input DFA $\DDD$ that is minimum for a Wheeler language $\LLL(\DDD)$.
\end{lemma}
\begin{proof}
    Let $\MMM=(Q_\MMM, \Sigma, \delta_\MMM, q_\MMM^0, F_\MMM)$ be the unique minimum Wheeler DFA recognizing $\LLL(\DDD)$, and let $n_\MMM$ be the number of its states. 
    Assume, for a contradiction, that the algorithm does not terminate. Since each iteration of the \texttt{while} loop in line~\ref{line:pop} removes a transition from $\Delta$, non-termination implies that $\Delta$ is replenished infinitely via calls to \fcopy{}. Because each such call adds a new state to $\LEX$, the number of states $|\LEX|$ must eventually exceed $n_\MMM$. Consider the iteration when $|\LEX| = n_\MMM + 1$ and let $\LEX = [v_0, v_1, \dots, v_{n_\MMM}]$. By Invariant~\ref{Loop Invariant: Connectivity}, each state $v_i$ is reachable from $q_0$ in $\mathcal{D}_w$, implying that the set $I_{\DDD_w}(v_i) = \{ \alpha \in \Sigma^* : \delta_w(q_0, \alpha) = v_i \}$ is non-empty for every $i\in \{0,\ldots, n_\MMM\}$. Hence, there exist $\alpha_i\in I_{\DDD_w}(v_i)$ for every $i\in \{0,\ldots, n_\MMM\}$. Furthermore, by Invariant~\ref{Loop Invariant: wheeler order} the DFA $\DDD_w$ is Wheeler and thus Lemma~\ref{lem: wheeler order} yields that $\alpha_i \prec \alpha_j$ for all $i<j$.
    As $\DDD_w^\Delta$ is co-accessible due to Invariant~\ref{Loop Invariant: co-accessible}, it follows that $\alpha_i\in \Pref(\LLL(\DDD_w^\Delta))$ for all $i\in \{0,\ldots, n_\MMM\}$. As $\LLL(\DDD_w^\Delta) = \LLL(\DDD)$ due to Invariant~\ref{Loop Invariant: language}, we furthermore have $\alpha_i\in \Pref(\LLL(\DDD))$. Hence also the minimum Wheeler DFA $\MMM$ for $\LLL(\DDD)$ has to contain states $w_0,\ldots, w_{n_\MMM}$ such that $\alpha_i\in  I_{\MMM}(w_i)$. Furthermore $\alpha_0 \prec \ldots \prec \alpha_{n_w}$ and Lemma~\ref{lem: wheeler order} imply that $w_0 \le_\MMM \ldots \le_\MMM w_{n_\MMM}$, where $\le_\MMM$ is the unique Wheeler order on $\MMM$. As $\le_\MMM$ is a total order and $\MMM$ has $n_\MMM$ states, we must have $w_{i - 1} = w_{i}$ for some $i\in [n_\MMM]$. Wheeler axiom W1 implies that $\MMM$ is input consistent and as $\delta_\MMM(q_\MMM^0,\alpha_{i - 1})=\delta_\MMM(q_\MMM^0, \alpha_i)$, the two strings $\alpha_{i-1}$ and $\alpha_i$ have to end with the same letter. As also $\delta_w(q_0, \alpha_{i-1})= v_{i-1}$ and $\delta_w(q_0, \alpha_{i})= v_{i}$, this implies $\lambda_{\DDD_w}(v_{i-1}) = \lambda_{\DDD_w}(v_i)$.
    Invariant~\ref{Loop Invariant: min order} now implies that $v_{i - 1} \not \equiv v_i$ in $\DDD_w^\Delta$. Hence there exists $\alpha\in\Sigma^*$ such that $|\{\alpha_{i - 1}\alpha, \alpha_i\alpha\} \cap \LLL(\DDD_w^\Delta)|=1$. Notice however that $\delta_\MMM(q_0, \alpha_{i - 1}\alpha) = \delta_\MMM(q_0, \alpha_{i}\alpha)$ and hence $|\{\alpha_{i - 1}\alpha, \alpha_i\alpha\} \cap \LLL(\MMM)|\in \{0, 2\}$. As $\LLL(\MMM)=\LLL(\DDD)=\LLL(\DDD_w^\Delta)$, this is a contradiction. Hence, the algorithm terminates.
\end{proof}

\paragraph*{Correctness.}
Having established termination, we prove that the resulting automaton is indeed the minimum Wheeler DFA for the target language.

\begin{lemma}[Correctness]
    \label{lemma: correctness refactored}
    The automaton $\DDD_w$ returned by \cref{algorithm: min wheeler} is the minimum Wheeler DFA recognizing $\LLL(\DDD)$.
\end{lemma}
\begin{proof}
    Recall that $\DDD_w = (\LEX, \Sigma, \delta_w, q_0, F_w)$ and $\DDD_w^\Delta = (\LEX, \Sigma, \delta_w \cup \Delta, q_0, F_w)$. At termination, since $\Delta$ is empty, then the algorithm returns $\DDD_w = \DDD_w^\Delta$. Invariant~\ref{Loop Invariant: Deterministic} implies that $\DDD_w$ is deterministic and Invariants~\ref{Loop Invariant: Connectivity} and~\ref{Loop Invariant: co-accessible} imply that  $\DDD_w = \DDD_w^\Delta$ is both accessible and co-accessible and thus trimmed. Invariant~\ref{Loop Invariant: language} furthermore implies that $\LLL(\DDD_w)=\LLL(\DDD)$ and from Invariant~\ref{Loop Invariant: wheeler order} it follows that $\DDD_w$ is Wheeler and $\LEX$ encodes its Wheeler order. Hence, $\DDD_w$ is a trimmed WDFA that accepts $\LLL(\DDD)$. \cref{thm: wheeler minimality} thus yields that $\DDD_w$ is the minimum  Wheeler DFA for $\LLL(\DDD)$ if and only if, for any two subsequent states $u$ and $v$, in its Wheeler order (that is encoded by $\LEX$), $u \equiv v$ implies $\lambda_{\DDD_w}(u) \neq \lambda_{\DDD_w}(v)$. This is exactly what Invariant~\ref{Loop Invariant: min order} states and hence this concludes the proof.
\end{proof}

\paragraph*{Runtime.} We now describe the data structures employed by Algorithm~\ref{algorithm: min wheeler} to construct the minimum Wheeler DFA $\DDD_w = (Q_w,\Sigma, \delta_w, p_0, F_w)$ within $O(m_w\log m_w)$ time. 

The algorithm maintains the following data structures: a read-only representation of the input DFA $\DDD$ supporting navigation queries in time $O(\log m) \subseteq O(\log m_w)$ (an adjacency list representation where the child labeled $a\in \Sigma$ of a given node $u$ can be found by binary search on the adjacency list of $u$); a queue $\Delta$ supporting push and pop operations in constant time; 
a dynamic set $F_w$ (a self-balancing tree);
a simple resizable array implementing $[\cdot]_{\cong}$;
a dynamic data structure supporting updates and queries on $\DDD_w$ (this structure will include $\LEX$, see below). 
We represent $\DDD_w$ using the same approach of Alanko et al.~\cite{alanko2020regular}. We sketch the overall idea next and give all implementation details at the end of this section.

The representation of $\DDD_w$ leverages on the following classic representation of Wheeler automata \cite{GAGIE201767}: 
any WDFA $\DDD_w$ with $n_w$ states  and $m_w$ transitions can be reconstructed from the four sequences, 
sorted in Wheeler order, of the nodes' names $\LEX$, 
incoming labels $\lambda' = \lambda(\LEX[1]), \dots, \lambda(\LEX[n_w])$ (i.e. $\lambda'[i]$ is the incoming label of $\LEX[i]$), 
in-degrees \IN\ (i.e. \IN$[i]$ is the in-degree of $LEX[i]$), and  outgoing labels \OUT\ (i.e. \OUT$[i]$ is the string formed by all distinct characters labeling the out-going transitions of $LEX[i]$). From such sequences, it is possible to reconstruct the original WDFA by exploiting 
the following fact:
\begin{lemma}[\cite{GAGIE201767}]\label{lem:bipartite}
	Consider the following total orderings of the transitions $\delta_w$ of $\DDD_w$:
	\begin{enumerate}
		\item Sort the transitions $(u,a,v)$ by the Wheeler order of their destinations (i.e. by $\LEX^{-1}[v]$), breaking ties by the Wheeler order of their sources (i.e. $\LEX^{-1}[u]$).
		\item Sort the transitions $(u,a,v)$ by the Wheeler order of their sources (i.e. $\LEX^{-1}[u]$), breaking ties by their label $a$.
	\end{enumerate} 
	Then, for any $a\in \Sigma$, the relative order of transitions labeled $a$ in the orderings is the same.
\end{lemma}

Suppose we now want to compute $\delta_w(u,a)$ on such a representation $(\LEX, \lambda', \IN, \OUT)$ (an operation which suffices to reconstruct the WDFA). 
First, we locate the position $i =\LEX^{-1}[u]$ of $u$ in Wheeler order. 
If $a \notin \OUT[i]$ then $\delta_w(u,a) = \bot$ and we are done.
Otherwise, we count the number $k$ of nodes in $\LEX[1,\dots, i]$ having an out-going edge labeled $a$. 
In other words, since $\DDD_w$ is deterministic, $(u,v,a)$ is the $k$-th transition labeled with $a$ in the total order (2) of Lemma \ref{lem:bipartite}. 
At this point, Lemma \ref{lem:bipartite} tells us that $(u,v,a)$ is the $k$-th transition labeled with $a$ in the total order (1) as well. We can therefore identify $v$ easily using $\lambda'$, \IN, and $\LEX$. 

Below we show how to dynamically maintain $(\LEX, \lambda', \IN, \OUT)$ so that a wide range of updates and queries on those sequences can be performed in $O(\log m_w)$ time each. At this point, it is not hard to show that each of the updates and queries on $\DDD_w$ performed by  Algorithm~\ref{algorithm: min wheeler} can be implemented in $O(\log m_w)$ time by reducing them to updates and queries on  $(\LEX, \lambda', \IN, \OUT)$ (see below).\ The claimed complexity of Algorithm~\ref{algorithm: min wheeler} follows immediately. Computing the spanning tree in Line \ref{line:SortedSpanningTree}, as well as performing the simple operations in Lines \ref{line:set Fw}-\ref{line: init congr MN} takes linear $O(m) \subseteq O(m_w)$ time. At this point note that, in each iteration of the main \texttt{while} loop, at least one new transition is added to $\DDD_w$ (Line \ref{line:insert transition}), while transitions are never deleted from $\DDD_w$. Each operation in the \texttt{while} loop takes $O(\log m_w)$ time (including the \texttt{foreach} operation at Line \ref{line:foreach2} --- see below for details, all those renamings of transitions are performed implicitly with $O(1)$ updates to $\IN$ and $\lambda'$), except the calls to \texttt{copy} in Lines \ref{line:CopyState1} and \ref{line:CopyState2}. Each call to \texttt{copy} may insert several new transitions in $\Delta$ and new states and transitions in $\DDD_w$. The former (new transitions in $\Delta$) will be processed in later iterations of the \texttt{while} loop and will be charged to the creation of a new transition of $\DDD_w$ each; the latter (new states and transitions in $\DDD_w$) amortize globally to $O(m_w\log m_w)$ time since states and transitions are never removed from $\DDD_w$.

\paragraph*{Implementing Dynamic Data Structures.}

We describe how all queries and updates on $\DDD_w$ performed by Algorithm \ref{algorithm: min wheeler} are implemented via queries and updates on the representation $(\LEX, \lambda', \IN, \OUT)$ of $\DDD_w$.

\subparagraph*{Data structures.}

$\LEX$ and $\lambda'$ are implemented with the dynamic string data structure 
of Nekrich and Navarro~\cite{NavarroN14},
representing any sequence $S$ over an integer alphabet in 
$O(|S|)$ words of space and supporting the following operations in $O(\log|S|)$ time:
\begin{itemize}
	\item Access any element: $S[i]$;
	\item Replace any character: given symbol (integer) $c$ and position $i$, set $S[i] \coloneq c$;
	\item $S.\mathtt{select}_c(i)$: position of the $i$-th symbol equal to $c$ in $S$;
	\item Insert a new symbol in an arbitrary position of $S$.
\end{itemize}
Observe that, since $\LEX$ is always a permutation, operation $\LEX^{-1}[q]$ is then solved simply as $\LEX.\mathtt{select}_q(1)$.
$\IN$ is represented using the dynamic searchable partial sum data structure of \cite{Prezza17}, using $O(|\IN|)$ words of space and supporting the following operations in $O(\log|\IN|)$: 
\begin{itemize}
	\item Partial sum: compute $\sum_{i=1}^j \IN[i]$ for any given $j \in \{1, \dots, |\IN|\}$;
	\item $\IN.\texttt{search}(k)$: given integer $k \ge 0$, return the minimum position $j$ such that $\sum_{i=1}^j \IN[i] \ge k$ (if any; otherwise, return $|\IN|+1$);
	\item Insert a new integer in an arbitrary position of $\IN$;
	\item Update: given any integers $i \in \{1, \dots, |\IN|\}$ and $\Delta$, update $\IN[i] \coloneq \IN[i] + \Delta$. 
\end{itemize}
Finally, $\OUT$ is a sequence of sequences. Let $k=|\OUT|$. We concatenate all those sequences in a dynamic string $\OUT' = \OUT[1] \OUT[2] \cdots \OUT[k]$ represented using the data structure of Nekrich and Navarro~\cite{NavarroN14}. We also keep a dynamic bitvector $B$ (again using the data structure of Nekrich and Navarro~\cite{NavarroN14}) storing the length of those strings in unary. Letting $t_i = |\OUT[i]|$, the bitvector is $B = 10^{t_1}10^{t_2}, \dots, 10^{t_k}$. At this point, 
it is not hard to see that
we can solve the following queries and updates in $O(\log|\OUT|)$ time on $\OUT$ by reducing them to queries and updates on $\OUT'$ and $B$ (we omit the details of such a classic reduction; see, e.g., \cite{alanko2020regular}): 
\begin{itemize}
	\item $\OUT.rank_c(i)$: given a character $c$ and a position $i \in [k]$, count the number of symbols equal to $c$ in the strings $\OUT[1, \dots, i]$.
	\item $\OUT.select_c(i)$: given a character $c$ and an integer $i\ge 1$, return the position $j$ such that $\LEX[j]$ is the $i$-th state in Wheeler order having an out-going transition labeled $c$. Return $0$ if no such state exists.
	\item  Given a character $c$ and a position $i \in [k]$, append $c$ to $\OUT[i]$.
\end{itemize}

Since all manipulated sequences have length $O(m_w)$, all basic operations discussed above, as well as the more complex operations described below, cost $O(\log m_w)$ time.

\subparagraph*{Creating New States: $\mathtt{insert}(q,j)$.} This update easily translates to an insert operation in all four components $\LEX, \lambda', \IN, \OUT$: we insert $q$ in position $j$ of $\LEX$;
we insert the special symbol $\#$ in position $j$ of $\lambda'$, signaling that $q$ is created (temporarily) without incoming edges;
we insert $0$ in position $j$ of $\IN$, signaling that the in-degree of $q$ is (temporarily) 0; we insert the empty string $\epsilon$ in position $j$ of $\OUT$, since $q$ (temporarily) has no out-going edges.

\subparagraph*{Inserting Transitions Into $\delta_w$.}

Adding individual transitions $(u,a,v)$ to $\delta_w$  (Lines \ref{line:modify T} and \ref{line:insert transition}) translates to 
the following operations. We compute $i_u = \LEX^{-1}[u]$ and  $i_v = \LEX^{-1}[v]$; append $a$ to $\OUT[i_u]$; increment $\IN[i_v] \coloneq \IN[i_v]+1$; and if $\lambda'[i_v] = \#$, we substitute $\lambda'[i_v] \coloneq a$.
	
\subparagraph*{Evaluating Transition Function.}
	
Evaluating the destination $v$ of $\delta_w(u,a)$ for any given state $u$ and character $a$ requires the following operations. We compute $i_u = \LEX^{-1}[u]$; 
if $\OUT[i_u]$ does not contain $a$ (we can discover this with two simple rank queries on $\OUT$), then we return $v=\bot$; otherwise, 
we compute the number $t$ of states before $u$ included (in Wheeler order) having an out-going transition labeled with $a$: $t = \OUT.rank_a(i_u)$; we compute the cumulative in-degrees $z$ of states with an incoming label strictly smaller than $a$: $z = \sum_{i=1}^{\lambda'.select_a(1)-1} \IN[i]$; we obtain (by Lemma \ref{lem:bipartite}) the position $i_v$ in Wheeler order of $v$: $i_v = \IN.search(z+t)$; we return $v = \LEX[i_v]$.

\subparagraph*{$\mathtt{pred}(u,a)$ and $\mathtt{succ}(u,a)$}
	
We discuss only $\mathtt{pred}(u,a)$, as $\mathtt{succ}(u,a)$ is symmetric. To solve $\mathtt{pred}(u,a)$ we proceed as follows. We retrieve $i_u = \LEX^{-1}[u]$; we compute the number $t$ of states before $u$ excluded  (in Wheeler order) having an out-going transition labeled with $a$: $t = \OUT.rank_a(i_u-1)$ (if $i_u=1$ or $t=0$, then we return $\bot$); finally, we return the $t$-th state in Wheeler order having an out-going transition labeled with $a$, that is, $\LEX[\OUT.select_a(t)]$.

\subparagraph*{Batch of Transitions Renamings (Line \ref{line:foreach2} of Algorithm).} 
	
The goal of this operation is to change the destination of all incoming transitions $(x,a,p)\in \delta_w$ of $p$ such that $u<x$ to $s$, i.e. renaming those transitions to $(x,a,s)$. 	
Even though this operation renames the destination of several transitions of $\DDD_w$, we can perform the whole batch of transitions with $O(1)$ updates to $\IN$ and $\lambda'$. This is possible thanks to (i) our WDFA representation, (ii)  to 
the fact that those transitions originally end in the same state $p$, and (iii) to the fact that after the renaming they end up in two adjacent (in Wheeler order) states $p<s$.

Recall that, at this point of execution, the in-degree of $s$ is $0$. 
Let $i_p = \LEX^{-1}[p]$ and $i_s = i_p+1$. In particular, $\LEX[i_s]=s$. 
Let $k = \IN[i_p]$. Let $d$ be the number of such transitions  $(x,a,p)\in \delta_w$ whose destination has to be renamed to $s$; the value $d$ can be found as follows.
Let $y$ be the largest state in Wheeler order such that $(y,a,p) \in \delta_w$, and let $i_y = \LEX^{-1}[y]$.  
Let moreover $i_u = \LEX^{-1}[u]$.
Then, $d = \OUT.rank_a(i_y) - \OUT.rank_a(i_u)$. We are left to find $i_y$. Letting $k$ being the cumulative number of incoming transitions entering in states less than or equal to $p$ in Wheeler order, then $i_y = \OUT.select_a(k)$. 
Finally, $k = \sum_{i=1}^{i_p} \IN[i] - \sum_{i=1}^{f_a-1} \IN[i]$, where $\LEX[f_a]$ is the first node in Wheeler order having an incoming transition labeled with $a$. Integer $f_a$ can be maintained explicitly for every $a$ using, for example, a self-balancing tree (even though this is not necessary as one can compute it with operations on $\OUT$ and $\IN$).

We are only left to perform the actual update, which requires just three operations: (1)  
$\lambda'[i_s] \coloneq a$, (2)
$\IN[i_p] \coloneq \IN[i_p] - d$, and (3) $\IN[i_s] \coloneq \IN[i_s] + d  = d$.

\subparagraph*{Operation at Line \ref{line:find i} of Algorithm.}
	
This operation boils down to finding the minimum alphabet's character $x$ being larger than or equal to $a$ and labeling some transition in $\delta_w$ (a simple self-balancing tree storing those characters can be used here). 
Then, $\lambda'.select_x(1)-1$ is the position in $\LEX$ of the largest node in Wheeler order having an incoming transition labeled with a character strictly smaller than $a$.

\section{Experiments}

We implemented our minimum WDFA algorithm, \minWheeler (Algorithm 1), in C++, the source code is available at \url{https://github.com/regindex/Minimum-WDFA-Constructor}. To evaluate its performance, we compared it with GCSA \cite{GCSA}, which, to the best of our knowledge, is the only other available tool capable of computing a Wheeler DFA (termed prefix-sorted automaton in the original work) for a Wheeler language, crucially without providing any minimality guarantee on the size of the output automata. GCSA provides a comprehensive pipeline that starts with a genomic variation file (.VCF) and a reference genome, and produces a final index that supports exact path-matching queries on the resulting pangenome graph. For a fair comparison with \minWheeler, we only ran the part of the GCSA pipeline that computes a WDFA from the input DFA, storing the pangenome to be indexed. For our test data, we employed the 23 automata used in the original GCSA benchmarks \cite{HelsinkiData} (Table 2 in \cite{GCSA}), encoding the Finnish subset of frequent mutations from the dbSNP database relative to the human reference chromosomes. Since GCSA operates on the reverse-deterministic automaton, we feed \minWheeler the reversed input, so that the two tools compute Wheeler DFAs on the same orientation and their output sizes can be directly compared.
All experiments were run on a server equipped with an Intel Xeon W-2245 CPU (3.90 GHz, 8 cores) and 128 GB of RAM, running Ubuntu 18.04 LTS 64-bit.

\begin{table*}[ht!]
\caption{Performance and other statistics. Comparison between GCSA and \minWheeler DFA construction. Time is expressed in minutes, Space in GB. \% increase is calculated as the relative increase of nodes in the Wheeler DFA compared to the initial graph nodes.}
\label{tab:chrom_stats}
\scalebox{0.85}{
\begin{minipage}{\textwidth}
\begin{tabular}{lrrrrrrrrr}
\toprule
 \multicolumn{3}{c}{ Input DFA }
 & \multicolumn{2}{c}{ GCSA } 
 & \multicolumn{2}{c}{ \minWheeler } 
 & \multicolumn{3}{c}{ Minimum  WDFA Size } \\
\cmidrule(lr){1-3} \cmidrule(lr){4-5} \cmidrule(lr){6-7} \cmidrule(lr){8-10}
Chrom.\ & Nodes & Edges
& Time & Space 
& Time & Space 
& Nodes & Edges & \% increase \\
\midrule
Chr 1  &  250.2M & 251.2M & 26.1 min & 17.1 GB & 24.4 min & 66.7 GB & 273.6M & 275.1M & 9.43\% \\
Chr 2  &  244.3M & 245.3M & 23.8 min & 17.1 GB & 31.3 min & 63.1 GB & 274.5M & 276.4M & 12.52\% \\
Chr 3  &  198.9M & 199.8M & -- & -- & -- & -- & $>$2200M & $>$2200M & $>$1100\% \\
Chr 4  &  192.1M & 193.0M & 19.8 min & 21.6 GB & 31.1 min & 49.7 GB & 226.2M & 228.4M & 18.05\% \\
Chr 5  & 181.7M & 182.5M & 18.2 min & 13.2 GB & 35.0 min & 47.0 GB & 216.2M & 217.9M & 19.19\% \\
Chr 6  &  172.0M & 172.8M & -- & -- & -- & -- & $>$2200M & $>$2200M & $>$1200\% \\
Chr 7  &  159.9M & 160.6M & 15.6 min & 11.6 GB & 31.7 min & 41.3 GB & 188.9M & 190.7M & 18.44\% \\
Chr 8  &  147.0M & 147.7M & -- & -- & -- & -- & $>$2200M & $>$2200M & $>$1400\% \\
Chr 9  &  141.8M & 142.3M & 14.0 min & 9.7 GB  & 11.4 min & 37.8 GB & 153.2M & 153.9M & 8.10\% \\
Chr 10 &  136.2M & 136.8M & 12.8 min & 9.5 GB  & 13.9 min & 35.2 GB & 151.4M & 152.3M & 11.25\% \\
Chr 11 &  135.6M & 136.3M & -- & -- & -- & -- & $>$2200M & $>$2200M & $>$1600\% \\
Chr 12 &  134.5M & 135.1M & 14.1 min & 10.8 GB  & 33.0 min & 34.8 GB & 177.9M & 180.9M & 33.09\% \\
Chr 13 &  115.6M & 116.1M & 11.2 min  & 8.1 GB  & 10.2 min  & 31.1 GB & 127.0M & 127.7M & 9.93\% \\
Chr 14 &  107.7M & 108.8M & 10.1 min & 7.2 GB  & 7.3 min  & 28.8 GB & 113.5M & 113.8M & 4.99\% \\
Chr 15 &  103.0M & 103.3M & 9.8 min  & 7.1 GB  & 9.4 min  & 27.5 GB & 112.9M & 113.5M & 9.74\% \\
Chr 16 &  90.7M & 91.1M  & -- & -- & -- & -- & $>$2200M & $>$2200M & $>$2000\% \\
Chr 17 &  81.5M  & 81.9M & 58.9 min & 117.2 GB & 609.8 min & 48.1 GB & 1205.3M & 1244.8M & 1399.5\% \\
Chr 18 &  78.4M  & 78.8M & -- & -- & 241.5 min & 29.0 GB & 654.7M & 700.8M & 762.28\% \\
Chr 19 &  59.4M  & 59.7M & 5.4 min  & 4.5 GB  & 8.5 min  & 15.4 GB & 70.0M  & 70.6M & 18.05\% \\
Chr 20 &  63.3M  & 63.6M & 5.6 min  & 4.6 GB  & 6.1 min  & 16.4 GB & 69.9M  & 70.3M & 10.48\% \\
Chr 21 &  48.3M  & 48.5M & 4.1 min  & 3.4 GB  & 3.1 min  & 12.5 GB & 51.6M  & 51.8M & 6.82\% \\
Chr 22 &  51.5M  & 51.7M & 4.3 min  & 3.6 GB  & 3.4 min  & 13.9 GB & 55.0M  & 55.2M & 6.78\% \\
Chr X  &  155.7M & 156.1M & 14.1 min & 10.6 GB & 14.3 min & 40.2 GB & 168.6M & 169.3M & 8.37\% \\
\bottomrule
\end{tabular}
\end{minipage}
}
\end{table*}

Table \ref{tab:chrom_stats} shows a summary of our results. From left to right, we report the reference chromosomes, followed by the numbers of nodes and edges for the corresponding input DFAs. We then show the wall-clock time and peak memory (RSS) achieved by the two competing software. Finally, we report the size of the minimum Wheeler DFA (WDFA) computed by \minWheeler and the resulting percentage increase in size. Dashes indicate cases where the computation did not finish, due to either exceeding the internal memory limit or reaching a 24-hour timeout. In those cases, we report the number of nodes and edges reached by the output WDFA at termination. On every chromosome where GCSA terminates, the WDFA in output has the same size as the minimum WDFA computed by \minWheeler, as the near-linear structure of pangenome graphs makes the number of state insertions needed for prefix-sortability essentially independent of the two algorithms.
We observe that both software do not terminate when the resulting WDFAs are very large. This behavior is expected, as the algorithms are output-sensitive; thus, their running time scales with the size of the final automata. A clear example is Chromosome 16, where the WDFA under construction reached 2.2 billion nodes within the first 24 hours, resulting in a size increase of over 2000\% relative to the input DFA. This indicates that, even if all input DFAs are acyclic and thus recognize a Wheeler language, the size of the minimum  WDFA recognizing such a language is not necessarily small.

In terms of performance, \minWheeler generally requires more resources than GCSA; specifically, on 16 input instances, our software uses between 2 and 4 times more memory. However, thanks to the incremental construction, which processes one transition at a time, \minWheeler does not suffer from the memory peaks observed in the competitor. In all cases where our computation did not finish, it terminated due to reaching the timeout, whereas GCSA always failed due to exceeding the internal memory size. This is evident for chromosomes 17 and 18: for chromosome 17, our software reaches a significantly lower memory peak, while for chromosome 18, it successfully completes construction, whereas GCSA fails.
Regarding running time, our software has proved competitive with GCSA. In particular, GCSA was at least two times faster in only 3 cases; while \minWheeler showed better or equivalent running time in 10 other instances, resulting in an average throughput of more than $10^5$ transitions/second. This is especially remarkable considering that \minWheeler relies on dynamic data structures to maintain the output WDFAs during construction.

In conclusion, we presented a proof-of-concept implementation of an algorithm that computes a minimum WDFA for a Wheeler language. Our experimental results show that our software is competitive with GCSA while providing the additional guarantee of minimum  output size. Additionally, we proved more robust in corner cases, successfully completing the WDFA construction of one additional chromosome. This confirms that our algorithm is effective at computing minimum WDFAs for sparse graphs, unlocking the construction of provably optimal-size pattern-matching data structures for pangenomes. On the other hand, our implementation relies on dynamic data structures, which can limit practical performance on large automata. Future work will focus on developing optimized data structures to further improve performance for large-scale applications.

\bibliography{WheelerDFA}

\newpage
\appendix

\section{Proof of Lemma~\ref{lemma: invariants refactored}}
\label{sec: invariant_proofs}
\ArxivOrCr{In order to prove Lemma~\ref{lemma: invariants refactored}, we first need three simple statements.}
{In order to prove Lemma~\ref{lemma: invariants refactored}, we need the following three simple statements whose proof can be found in the full version of the article~\cite{arxiv}.}
\begin{restatable}{lemma}{movetransition}\label{lem: modify}
    Let $\DDD = (Q, \Sigma, \delta, q_0, F)$ be a DFA. 
    \begin{enumerate}[(1)]
        \item \label{item: move transition} If we modify $\DDD$ by replacing a transition $\delta(u, a) = v$ with $\delta(u, a) = v'$, for any $v'\equiv v$, then the equivalence relation $\equiv$ and the language $\LLL(\DDD)$ remain unchanged. 
        \item \label{item: copy out-transitions} Let $u\in Q$. If we modify $\DDD$ by adding a new state $u'$ to $Q$, adding a new transition $\delta(u',a)=v$ for every transition $\delta(u,a)=v$, and making $u'$ final if and only if $u$ is final, then the language $\LLL(\DDD)$ does not change, and $u \equiv u'$ in the modified DFA. Apart from the addition of $u'$, the relation $\equiv$ does not change.
        \item \label{item: move in transition} Assume $\DDD$ to be co-accessible and let $q, q'\in Q$ be  distinct such that $q\equiv q'$. If we modify $\DDD$ by replacing a transition $(p, a, q)$ with a transition $(p, a, q')$, then $\DDD$ stays co-accessible.
    \end{enumerate}
\end{restatable}
\ArxivOrCr{
\begin{proof}
    \begin{enumerate}[(1)]
        \item  Let us denote with $\delta'$ the transition function obtained from $\delta$ by replacing the tuple $(u,a,v)$ with $(u, a, v')$ for some $v'\equiv v$. Recall that $p\equiv p'$ holds for $p, p'\in Q$ if and only if, for all $\alpha \in \Sigma^*$, we have $\delta(p, \alpha) \in F$ if and only if $\delta(p', \alpha) \in F$. 
        Observe also that $\LLL(\DDD) = \{\alpha\in\Sigma^*: \delta(q_0, \alpha) \in F\}$. Hence, it is enough to show that for any $q\in Q$ and $\alpha \in \Sigma^*$, we have $p=\delta(q, \alpha)\in F$ if and only if $p'=\delta'(q, \alpha)\in F$.
        Let $q\in Q$ be arbitrary. We show the statement, for any $\alpha\in \Sigma^*$, it holds that $p=\delta(q, \alpha)\in F$ if and only if $p'=\delta'(q, \alpha)\in F$ by induction over the length of $\alpha\in \Sigma^*$. 
        
        First assume that $|\alpha|=0$ and thus $\alpha=\epsilon$. Then $p=p'=q$ and hence the statement holds trivially. Now assume that $|\alpha|>1$ and write $\alpha = a \beta$. By induction, we have that, for any $q'\in Q$, $\delta(q', \beta)\in F$ if and only if $\delta'(q', \beta)\in F$. Let $r=\delta(q, a)$ and $r'=\delta'(q, a)$. We distinguish two cases. (1)~If $r=r'$, we have $p=\delta(\delta(q, a), \beta)=\delta(r, \beta)$ and $p'=\delta'(\delta(q, a), \beta)=\delta'(r, \beta)$ and hence the statements follows using the induction hypothesis. (2)~If $r\neq r'$, it holds that $r=v$ and $r'=v'$, i.e., the transition $(q,a,r)$ in $\delta$ was replaced by the transition $(q, a, r')$ in order to obtain $\delta'$. Then $r\equiv r'$, which means that, for all $\gamma\in \Sigma^*$, we have $\delta(r, \gamma)\in F$ if and only if $\delta(r', \gamma)\in F$. Thus, we have $p=\delta(r, \beta)\in F$ if and only if $\delta(r', \beta)\in F$. Using the induction hypothesis the latter happens if and only if $p'=\delta'(r', \beta)\in F$. This concludes the proof.
        \item Let $\DDD'=(Q', \Sigma, \delta', q_0, F')$ with $Q'=Q\cup\{u'\}$ and $\delta' := \delta \cup \{(u', a, v): (u, a, v)\in \delta\}$ be the DFA after the modifications and let $\equiv'\;:=\;\equiv \cup \{(u', v), (v, u'): (u, v)\in \equiv\}$. Recall that $\LLL(\DDD) = \{\alpha\in\Sigma^*: \delta(q_0, \alpha) \in F\}$ and that $p\equiv p'$ holds for two states $p, p'\in Q$ if and only if, for all $\alpha \in \Sigma^*$, we have $\delta(p, \alpha) \in F$ if and only if $\delta(p', \alpha) \in F$.
    
        Observe that, (1)~for all $q\in Q$ and $\alpha \in \Sigma^*$, we have $\delta(q, \alpha)= \delta'(q, \alpha)$, (2)~we have $\delta'(u', \alpha)=\delta(u, \alpha)$ for $\alpha\neq \epsilon$, and (3)~$\delta'(u', \epsilon)\in F$ if and only if $\delta(u, \epsilon)\in F$. 
        Now, (1)~implies that $\LLL(\DDD)=\LLL(\DDD')$. It remains to verify that $\equiv'$ is the Myhill-Nerode equivalence relation $\equiv_{\DDD'}$ on $\DDD'$. Let $q, q'\in Q'$. If $\{q, q'\}\subseteq Q$, (1)~implies that $q \equiv_{\DDD'} q'$ if and only if $q \equiv q'$, which by the definition of $\equiv'$ holds if and only if $q \equiv' q'$. If $\{q, q'\}\not\subseteq Q$, assume w.l.o.g.\ that $q'=u'$. We have $q\equiv' u'$ if and only if $q\equiv u$ using the definition of $\equiv'$. Finally, (1), (2) and (3) imply that $q\equiv u$ if and only if $q\equiv_{\DDD'}u'$.
        \item We prove the following statement by strong induction over $k>1$. States that could reach a final state through a path $P$ of length $k$ can still reach a final state through a path of length $k$ after the modification of $\DDD$. 

        For the induction base, assume $v$ to be a state that could reach a final state through a path $P$ of length $k=1$. Then $v\in F$ and there is nothing to show.
        Now assume $v$ to be a state that could reach a final state through a path $P$ of length $k > 1$ before the modification and assume the statement to be true for all $k'<k$. If the path $P$ does not go through the transition $(p, q, q)$, there is nothing to show as $P$ still exists after the modification. Hence, assume that $P = \pi, p, q, \rho$ for some subpaths $\pi$ and $\rho$ and assume $\pi$ to be of shortest length with that property. This means that $q, \rho$ is a path to a final state. Let $\alpha$ be its label. As $q\equiv q'$, it follows that there exists a path $\rho'$ labeled $\alpha$ leading to a final state. The path $q', \rho'$ is of length $k'=|q',\rho'|=|q,\rho|<k$ and thus by the induction hypothesis there exists a path, say $\rho''$, of length $k'$ such that $q',\rho''$ leads to a final state also after the modification. It now follows that $\pi, p, q', \rho''$ is a path from $v$ to a final state that exists after the modification. \qedhere 
    \end{enumerate}
\end{proof}
}{}

\ArxivOrCr{
We are now ready to prove all the invariants in Lemma~\ref{lemma: invariants refactored}. We restate the lemma for better readability.
\invariantsrefactored*
\medskip

We start by arguing that the invariants all hold initially, i.e., before the first execution of the while loop. Thereafter we argue that the invariants are maintained in each of the three cases that the algorithm considers (depending on the evaluation of the two if conditions within in the while loop).
}{\medskip

We are now ready to prove Lemma~\ref{lemma: invariants refactored}. We start by arguing that the invariants hold initially and then argue that they are maintained during each iteration.
}

\ArxivOrCr{\paragraph*{Initialization.}}{\subparagraph*{Initialization.}}
\ArxivOrCr{
    \begin{description}
        \item[Invariants~\ref{Loop Invariant: Deterministic}, \ref{Loop Invariant: wheeler order}, and \ref{Loop Invariant: language}] 
        These invariants hold trivially since $\DDD_w^\Delta = \DDD$ and $\DDD_w$ is by construction Wheeler with Wheeler order encoded by $\LEX$.
        
        \item[Invariants~\ref{Loop Invariant: Connectivity} and \ref{Loop Invariant: co-accessible}]
        By minimality of $\DDD$, every state in 
        $\DDD$ is reachable from $q_0$ and can reach a final state. As initially $\DDD_w^\Delta=\DDD$ this shows co-accessibility of $\DDD_w^\Delta$. Moreover, $\DDD_w$ is a spanning out-tree of $\DDD$ rooted in $q_0$, which shows accessibility of $\DDD_w$.

        \item[Invariants \ref{Loop Invariant: MN-equiv} and \ref{Loop Invariant: min order}] As $\DDD = \DDD_w^\Delta$ is minimum, no two distinct states of $\DDD_w^\Delta$ are Myhill-Nerode equivalent in $\DDD_w^\Delta$. This immediately implies Invariant~\ref{Loop Invariant: min order} and furthermore implies that Line~\ref{line: init congr MN} sets $\cong$ to be equal to $\equiv$, thus also implying Invariant~\ref{Loop Invariant: MN-equiv}.
    \end{description}
}{
    Invariants~\ref{Loop Invariant: Deterministic}, \ref{Loop Invariant: wheeler order}, and \ref{Loop Invariant: language} hold trivially since $\DDD_w^\Delta = \DDD$ and $\DDD_w$ is by construction Wheeler with Wheeler order encoded by $\LEX$. For Invariants~\ref{Loop Invariant: Connectivity} and \ref{Loop Invariant: co-accessible}, note that by minimality of $\DDD$, every state in $\DDD$ is reachable from $q_0$ and can reach a final state. As initially $\DDD_w^\Delta=\DDD$ this shows co-accessibility of $\DDD_w^\Delta$. Moreover, $\DDD_w$ is a spanning out-tree of $\DDD$ rooted in $q_0$, hence $\DDD_w$ is accessible. For Invariants \ref{Loop Invariant: MN-equiv} and \ref{Loop Invariant: min order}, observe that $\DDD = \DDD_w^\Delta$ is minimum and no two distinct states of $\DDD_w^\Delta$ are Myhill-Nerode equivalent. This immediately implies Invariant~\ref{Loop Invariant: min order} and also implies that Line~\ref{line: init congr MN} sets $\cong$ to be equal to $\equiv$, thus also Invariant~\ref{Loop Invariant: MN-equiv} holds.
}

\ArxivOrCr{\paragraph*{Maintenance.}}{\subparagraph*{Maintenance.}}

Every iteration of the \texttt{while} loop of Algorithm \ref{algorithm: min wheeler} removes a transition $(u, a, v)$ from $\Delta$. As in the algorithm, we let $p \coloneqq \delta(\texttt{pred}(u,a),a)$ and $s \coloneqq \delta(\texttt{succ}(u,a),a)$. The algorithm then modifies $\DDD_w$ and $\DDD_w^\Delta$ in one of the following three ways, depending on the branches taken in the two \texttt{if} statements in the while loop.
We show that the invariants are maintained in each of the three cases. Note that when proving that an Invariant $k$ is maintained we only use Invariants $l$ such that $l<k$.
 
\subparagraph*{Case 1: $v \cong p$ or $v \cong s$.}  Hence, the \texttt{if} condition at line \ref{line:if1} holds.
The algorithm's actions: 
\begin{enumerate}[(i)]
    \item The transition $(u,a,v)$ is removed from $\Delta$. \label{item: case 1 remove transition}
    \item The transition $(u,a,v')$ is added to $\delta_w$, where $v' \in \{p,s\}$ is such that $v' \cong v$. \label{item: case 1 add transition}
\end{enumerate}

\begin{description}
    \item[Invariant \ref{Loop Invariant: Deterministic}.] The invariant is trivially maintained.
    \item[Invariant \ref{Loop Invariant: Connectivity}.]  
    Only Action~\eqref{item: case 1 add transition} modifies $\DDD_w$. As it adds a transition between existing states, the invariant is trivially maintained.

    \item[Invariant \ref{Loop Invariant: wheeler order}.] 

    By its definition, $v'$ has at least one incoming transition in $\DDD_w$, hence $v'\neq q_0$ and 
    Invariant~\ref{Loop Invariant: Connectivity} implies that $q_0$ is still the only state with in-degree equal to zero in $\DDD_w$.
    
    (W1) Action~\eqref{item: case 1 add transition} adds $(u,a,v')$ to $\delta_w$
    and, by definition of $v'$, $\lambda_{\DDD_w}(v') = a$ before the action. 
   Hence $\lambda_{\DDD_w}$ does not change.
    Since $\LEX$ is also unchanged, $q_0$ still precedes all states in $\LEX$ after the action and the Wheeler property W1 is maintained.

    (W2) The actions only touch transitions labeled with $a$, so we can ignore those labeled with other letters. Action~\eqref{item: case 1 add transition} adds $(u,a,v')$ to $\delta_w$, so we have to verify that for all $x\in Q_w$ such that $\delta_w(x,a) = x'$ is defined, $x<u$ implies $x' \le v'$ (the case ``$u<x$ implies $v' \le x'$'' is completely symmetric and we do not treat it). Let $x$ be such that $\delta_w(x,a) = x'$ is defined and $x<u$; if no such state exists, then we are done. We therefore assume that $x$ exists. Note that $x < u$ implies that $\texttt{pred}(u,a)$ exists, therefore $p \coloneqq \delta(\texttt{pred}(u,a),a)$ is defined as well. 
    If $s$ is defined, then $\texttt{pred}(u,a) < \texttt{succ}(u,a)$ and property W2 (before the action) imply that $p \le s$. 
    Two cases can happen. (i) If $x = \texttt{pred}(u,a)$, then $x' = p \le s$ (if $s$ is defined; otherwise, $x' = p$) and we are done since $v' \in \{s,p\}$ implies that $x' \le v'$.
    (ii) If $x \neq \texttt{pred}(u,a)$, then it must be $x < \texttt{pred}(u,a)$ by the definitions of $\texttt{pred}(u,a)$ and $x$. Property W2 (before the action) implies $x' = \delta(x,a) \le \delta(\texttt{pred}(u,a),a) = p$. Also in this case we are done: $v' \in \{s,p\}$ and (if $s$ exists) $p\le s$ implies $x' \le v'$.
    
    \item[Invariants \ref{Loop Invariant: MN-equiv} and \ref{Loop Invariant: language}.] By assumption, before applying the actions $\equiv$ and $\cong$ are the same equivalence relation. After applying the actions, by Lemma~\ref{lem: modify}~\eqref{item: move transition} the equivalence relation $\equiv$ does not change. Since the actions do not modify $\cong$ nor the language, we conclude that $\equiv$ and $\cong$ are still the same relation, hence Invariants \ref{Loop Invariant: MN-equiv} and \ref{Loop Invariant: language} are maintained. 
    
    \item[Invariant \ref{Loop Invariant: min order}.] 
    By Invariant \ref{Loop Invariant: wheeler order}, the automaton $\DDD_w$ is Wheeler and hence input-consistent after the actions.
    Since the actions only add one transition to $\delta_w$, we conclude that $\lambda_{\DDD_w}(v)$ and $\lambda_{\DDD_w}(v')$ do not change. 
    No other state's incoming transitions are modified, so $\lambda_{\DDD_w}$ does not change.
    By Invariant \ref{Loop Invariant: MN-equiv}, the relation $\equiv$ does not change through the actions either. 
    Additionally, the actions do not modify $\LEX$.
    Hence Invariant \ref{Loop Invariant: min order} is maintained. 
    
    \item[Invariant \ref{Loop Invariant: co-accessible}.] Actions~\eqref{item: case 1 remove transition} and~\eqref{item: case 1 add transition} replace the transition $(u, a, v)$ in $\DDD_w^\Delta$ with the transition $(u, a, v')$, where $v\cong v'$. By Invariant~\ref{Loop Invariant: MN-equiv}, it holds that $v\equiv v'$ and hence Lemma~\ref{lem: modify}~\eqref{item: move in transition} yields that $\DDD_w^\Delta$ is still co-accessible after this transformation.
\end{description}

\subparagraph*{Case 2: $p \not\cong v \not\cong s$ and $p = s \neq \bot$.} In other words, the \texttt{if} condition at line \ref{line:if1} does not succeed and the \texttt{if} condition at line \ref{line:if2} succeeds. 
Summary of the algorithm's actions:
\begin{enumerate}[(i)]
    \item The transition $(u,a,v)$ is removed from $\Delta$. \label{item: case 2 remove uav}
    \item The state $p=s$ is ``split'' into two distinct states $p\neq s$ with the same $\cong$-class and same ``final'' status. This state $s$ is put in $\LEX$ in the position immediately following $p$. For each transition $(p,c,x)\in \Delta$, a new transition $(s,c,x)$ is inserted into $\Delta$. For each transition $(p,c,x)\in \delta_w$, a new transition $(s,c,x)$ is inserted into $\delta_w$. \label{item: case 2 copy p}
    \item Every $(x,a,p)\in \delta_w$ with $u<x$ (in Wheeler order), is replaced (in $\delta_w$) with $(x,a,s)$.\label{item: case 2 move transitions}
    \item A new state $v'$ is inserted between $p$ and $s$ (so that $p<v'<s$ are adjacent in $\LEX$), then every outgoing transition $(v, c, x)\in \Delta \cup \delta_w$ of $v$ is copied as $(v', c, x)$ into $\Delta$, and the $\cong$-class of $v'$ and its ``final'' status are set to be equal to those of $v$. \label{item: case 2 copy v}
    \item The transition $(u,a,v')$ is inserted into $\delta_w$. \label{item: case 2 insert u a vp}
\end{enumerate}

\begin{description}
    \item[Invariant \ref{Loop Invariant: Deterministic}.] The new states $s$ and $v'$ get the same outgoing transitions as $p$ and $v$, respectively. No other state gets a new outgoing label in $\DDD_w^\Delta$, hence the invariant is maintained.

    \item[Invariant \ref{Loop Invariant: Connectivity}.]  
    We show that all states $q$ in $\DDD_w$ are reachable from $q_0$ in $\DDD_w$ after the actions by distinguishing the following three cases: (1) $q$ is a previously existing state with $q\neq p$, (2) $q\in \{p, s\}$, and (3) $q=v'$.

    For (1), let $q\neq p$ be a previously existing state different from $p$ and let $P$ be a directed path from $q_0$ to $q$ in $\DDD_w$ before the actions. Such a path exists as by assumption every state is reachable from $q_0$ in $\DDD_w$ before the actions. Now, note that the only transitions removed from $\delta_w$ are the transitions $(x, a, p)$ such that $u<x$ through Action~\eqref{item: case 2 move transitions}. Hence, if $P$ does not contain any such transition, the path $P$ still exists after the actions and nothing is to be shown. Thus, assume $P$ to contain some transition $(x, a, p)$ and let $y$ be the successor of $p$ on the path $P$, say through a transition $(p, c, y)$. Action~\eqref{item: case 2 move transitions} inserts the transition $(x, a, s)$ into $\delta_w$ and Action~\eqref{item: case 2 copy p} inserts the transition $(s, c, y)$ into $\delta_w$. Hence the state $q$ is now reachable from $q_0$ through the path $P'$ that is obtained by replacing such transitions $(x, a, p), (p, c, y)$ with $(x, a, s), (s, c, y)$. This shows (1). 
    For (2), notice that $\pred(u, a)\neq\succs(u,a)$ and thus we cannot have $\pred(u, a) = p = \succs(u,a)$. We assume that $p\neq \pred(u, a)$, the case $p\neq \succs(u,a)$ is symmetric. In this case $\pred(u, a)$ is a previously existing state different from $p$ and by (1) $\pred(u, a)$ is reachable from $q_0$ and then so is $p$ via the transition $(\pred(u, a), a, p)$. It remains to argue that $s$ is reachable from $q_0$. Notice that the transition $(\succs(u, a), a, s)$ is contained in $\delta_w$ after Action~\eqref{item: case 2 move transitions} and thus it suffices to argue that $\succs(u, a)$ is reachable from $q_0$. If $\succs(u, a)=p$ this follows from the previous argument, if $\succs(u, a)\neq p$, it follows from (1).
    For (3), note that the new state $v'$ is reachable from $q_0$ via a path to the state $u$ (that exists due to (1) and (2)) and the transition $(u, a, v')$ that is inserted through Action~\eqref{item: case 2 insert u a vp}.
    
    \item[Invariant \ref{Loop Invariant: wheeler order}.] 
    Clearly, $q_0$ does not get a new incoming transition in $\delta_w$ as the only states that get new in-transitions in $\delta_w$ are $v'$ and $s$. Furthermore, by Invariant~\ref{Loop Invariant: Connectivity}, all states are reachable from $q_0$ and hence all states different from $q_0$ have an in-transition. We conclude that $q_0$ is still the only state without incoming transitions in $\DDD_w$.

    (W1) The actions remove some (but not all) of the incoming transitions of $p$, hence $\lambda_{\DDD_w}(p)$ does not change. The two new states $v'<s$ are inserted immediately after $p$ and get new incoming transitions with label $a$. We conclude that W1 is maintained by the actions. 
    
    (W2) The actions only touch transitions labeled with $a$, so we can ignore those labeled with other letters. There are pre-existing transitions labeled with $a$ -- we call them of type (o), and three types of newly inserted transitions labeled with $a$ that the algorithm inserts through the actions: 
    (a)~Action~\eqref{item: case 2 copy p} inserts a transition $(s, a, x)$ into $\delta_w$ for every transition $(p, a, x)$ in $\delta_w$, 
    (b) Action~\eqref{item: case 2 move transitions} inserts a transition  $(x, a, s)$ into $\delta_w$ for every transition $(x, a, p)$ in $\delta_w$ with $u<x$, 
    (c)~Action~\eqref{item: case 2 insert u a vp} inserts the transition $(u, a, v')$ into $\delta_w$.
    
    Let $(y, a, y')$ and $(z, a, z')$ be two transitions after the actions, with $y<z$. We need to prove that $y'\le z'$. We distinguish cases based on the types (a, b, c, or o) of the two transitions. We denote these cases with (t1/t2), where (t1) is the type of transition $(y, a, y')$ and (t2) the type of transition $(z, a, z')$. This results in a total of 16 cases. 

   \ArxivOrCr{First assume that $(y, z, y')$ is of type (o). If also $(z, a, z')$ is of type (o) (Case (o/o)), there is nothing to show as the relative order of $y,z,y',z'$ in $\LEX$ are not changed. Now assume that $(z, a, z')$ is of type (a), i.e., $z=s$ (Case (o/a)). Recall that $v'$ and $s$ are inserted in this order subsequent to $p$ in $\LEX$. The assumption $y < z = s$ thus implies that $y\le p$ as none of the actions inserts an out-transition at state $v'$. If $y=p$, by determinism (Invariant~\ref{Loop Invariant: Deterministic}), we have $y'=z'$. Now assume $y<p$. The existence of  $(s, a, z')$ implies that there has been a transition $(p, a, z')$ previously, hence the invariant applied to the transitions $(y, a, y')$ and $(p, a, z')$ together with $y<p$ yields $y'<z'$. Now assume that $(z, a, z')$ is of type (b), i.e., $z'=s$ (Case (o/b)). The insertion of the transition $(z, a, s)$ implies that there was a transition $(z, a, p)$ previously. From $y < z$ we thus conclude $y'\le p\le s=z'$ as $s$ is inserted in $\LEX$ after $p$. Finally assume that $(z, a, z')$ is of type (c) (Case~(o/c)), we have $z=u$, $z'=v'$ and $y<z=u$. By definition $\pred(u, a)$ is the largest state in $\LEX$ preceding $u$ with outgoing label $a$. From $y<u$ it thus follows that $y \le \pred(u, a)$. If $y=\pred(u, a)$, it follows that $y'=p$ and hence $y'=p\le v'=z'$ as $v'$ is inserted after $p$ in $\LEX$. If instead $y<\pred(u, a)$, W2 before the actions for the two transitions $(y, a, y')$ and $(\pred(u, a), a, p)$ implies that $y'\le p$. As $v'$ is inserted in $\LEX$ after $p$, we get $y'\le v'=z'$.
    
    Now assume that $(y, z, y')$ is of type (a), i.e., $y=s$. This means that there was previously a transition $(p, a, y')$. If $(z, a, z')$ existed previously (Case~(a/o)), it follows that $p < s < z$ and thus $y'\le z'$. Case~(a/a) cannot occur as $y<z$ contradicts $y=z=s$. If $(z, a, z')$ is of type (b) (Case~(a/b)), we have $z'=s$ and the insertion of the transition $(z, a, s)$ implies that there was a transition $(z, a, p)$ previously. From $p < s = y < z$ we thus conclude $y'\le p\le s=z'$ as $s$ is inserted in $\LEX$ after $p$. If $(z, a, z')$ is of type (c) (Case~(a/c)), we have $z=u$, $z'=v'$ and $s<u$. The latter implies that $p<u$ as $p\le s$. By definition $\pred(u, a)$ is the largest state in Wheeler order smaller than $u$ with outgoing label $a$. From $p<u$ it thus follows that $p < \pred(u, a)$. Now, W2 before the actions for the two transitions $(p, a, y')$ and $(\pred(u, a), a, p)$ implies that $y'\le p$. As $v'$ is inserted in $\LEX$ after $p$, we get $y'\le v'$.

    Now assume that $(y, a, y')$ is of type (b), i.e., $y'=s$ and $y > u$. The insertion of $(y, a, s)$ implies that there has previously been a transition $(y, a, p)$ in $\delta_w$. If $(z, a, z')$ is of type (o) (Case~(b/o)), we thus get $y'=p\le z'$ by the invariant. If $(z, a, z')$ is of type (a), i.e., $z=s$ (Case~(b/a)), we have $z=s$ and $y<s$ and the insertion of $(s, a, z')$ implies that there has previously been a transition $(p, a, z')$ in $\delta_w$. Notice that $u < y < z = p$ implies that $p \le z'$. Recall that $p<v'<s$ are subsequent states in $\LEX$. If $z'=p$, we get a contradiction as a transition $(z, a, p)$ is deleted from $\delta_w$ by Action~\eqref{item: case 2 move transitions}. Hence $p < z'$. We also have $z'\neq v'$ as that would imply that $(z, a, z')$ is of type (c). Hence $z'\ge s =y'$.
    Case~(b/b) is trivial as $y'=s=z'$. 
    If $(z, a, z')$ is of type (c) (Case~(b/c)), we have $z=u$, $z'=v'$, and $y < z = u$. This case cannot occur as a transition $(y, a, s)$ of type (b) is inserted only for $y > u$, contradicting $y < u$. 
    
    Now assume that $(y, a, y')$ is of type (c), i.e., $y=u$, $y'=v'$, and $u < z$. Furthermore, assume that $(z, a, z')$ is of type (o) (Case~(c/o)). As $u < \succs(u, a)\le z$, the invariant yields $s\le z'$. This immediately implies $y'=v'\le z'$ as $v'$ is inserted before $s$ in $\LEX$. Now assume that $(z, a, z')$ is of type (a), i.e., $z=s$ (Case~(c/a)). It turns out that this case never occurs as the assumption $u=y<z=s$ yields that $u\le p$ as $p$ immediately precedes $s$ in $\LEX$. The definition of $p$ however entails that $p<u$. Now assume that $(z, a, z')$ is of type (b), i.e., $z'=s$ (Case~(c/b)). This case is trivial as $y'\le z'$ follows from $y'=v'$ being inserted before $z'=s$ in $\LEX$. Case~(c/c) cannot occur as there is a single transition of type (c) that gets introduced.}{
    Out of those 16 cases, 4 cases ((a/a), (b/c), (c/a), (c/c)) cannot occur (e.g., because there is a single transition of type (c) or because they contradict a previous assumption on the order, e.g., $p<u$ implies that (c/a) does not occur. Due to space limitations we omit the case distinction here and point the interested reader to the full version~\cite{arxiv}. 
    }
    \item[Invariants \ref{Loop Invariant: MN-equiv} and  \ref{Loop Invariant: language}.]
    By the invariants, $\cong$ and $\equiv$ are the same relation, and $\LLL(\DDD_w^\Delta) = \LLL(\DDD)$ before the actions. 
    Since we modify the structure of $\DDD_w^\Delta$ by performing only actions mentioned in Lemma~\ref{lem: modify}~\eqref{item:  move transition} and \eqref{item: copy out-transitions}, after the actions $s\equiv p$ and $v' \equiv v$ hold, the relation $\equiv$ does not change apart from the addition of $s$ and $v'$, and the language of $\DDD_w^\Delta$ does not change. Since Action~\eqref{item: case 2 copy p} sets $[s]_{\cong} \coloneqq [p]_{\cong}$ and Action~\eqref{item: case 2 copy v} sets $[v']_{\cong} \coloneqq [v]_{\cong}$, we also conclude that $\cong$ and $\equiv$ are the same relation after the actions.   
    
    \item[Invariant \ref{Loop Invariant: min order}.]
    Let $u'$ be a state adjacent to $p=s$ before the actions. By the invariant, if $u' \equiv p=s$, then $\lambda_{\DDD_w}(u') \neq \lambda_{\DDD_w}(p)$. 
    The actions split $p$ into two $\equiv$-equivalent adjacent states $p \neq s$ (variable $s$ is renamed) with $\lambda_{\DDD_w}(p) = \lambda_{\DDD_w}(s)$, and a new state $v' \equiv v$ is inserted between them: $p < v' < s$. State $u'$ is now adjacent to either $p$ or $s$; in the former case, it still holds that if $u' \equiv p$, then $\lambda_{\DDD_w}(u') \neq \lambda_{\DDD_w}(p)$. The latter case is analogous. 
    Finally, $v' \equiv v \not\equiv p \equiv s$ implies that the invariant is also maintained between $v'$ and its neighbors $p,s$.

    \item[Invariant~\ref{Loop Invariant: co-accessible}.] Let $\DDD_w^\Delta$ be the DFA before the actions were executed. Consider now the DFA that results from $\DDD_w^\Delta$ with only Action~\eqref{item: case 2 copy p} executed (in particular this DFA still contains the transition $(u, a, v)$). As all out-transitions of $p$ in $\DDD_w^\Delta$ are copied as out-transitions to $s$, it follows that $s$ can reach the same final state in this DFA as $p$. Now consider the DFA where in addition Action~\eqref{item: case 2 move transitions} was executed. We have $p\cong s$ and by Invariant~\ref{Loop Invariant: MN-equiv} $p\equiv s$. Lemma~\ref{lem: modify}~\eqref{item: move in transition} implies that after each replacement of a transition $(x, a, p)$ with a transition $(x, a, s)$, the DFA remains co-accessible. Now consider the DFA where in addition Action~\eqref{item: case 2 copy v} was executed. As all out-transitions of $v$ in $\DDD_w^\Delta$ are copied as out-transitions to $v'$, it follows that $v'$ can reach the same final state in this DFA as $v$. Finally, consider the DFA after all actions, i.e., assume that in addition the replacement of transition $(u, a, v)$ with transition $(u, a, v')$ was executed through Actions~\eqref{item: case 2 remove uav}~and~\eqref{item: case 2 insert u a vp}. Note that we have $v\cong v'$ and by Invariant~\ref{Loop Invariant: MN-equiv} $v\equiv v'$. Lemma~\ref{lem: modify}~\eqref{item: move in transition} then again implies that the DFA after these two actions is co-accessible.
    
\end{description}

\subparagraph*{Case 3: $p \not\cong v \not\cong s$ and $(p \neq s \vee s=\bot)$.} 
In other words, both the \texttt{if} conditions at lines \ref{line:if1} and \ref{line:if2} do not succeed. Summary of the algorithm's actions:
\begin{enumerate}[(i)]
    \item The transition $(u,a,v)$ is removed from $\Delta$. \label{item: case 3 remove uav}
    \item A new state $v'$ is created and inserted at position $i$ in $\LEX$, where the value of $i$ depends on whether $p=\bot$: if $p\neq \bot$ then $i = 1+\LEX^{-1}[p]$, else $i\coloneqq 1+ \max(\{j\ :\ \lambda(\LEX[j])\prec a\} \cup \{0\})$, i.e., $i-1$ is the position of the last state in $\LEX$ with incoming letter strictly smaller than $a$ or $0$ if all states in $\LEX$ have an incoming letter at least $a$.
    Then every outgoing transition $(v, c, x)\in \Delta \cup \delta_w$ of $v$ is copied as $(v', c, x)$ into $\Delta$, and the $\cong$-class of $v'$ and its ``final'' status are set to be equal to those of $v$. \label{item: case 3 copy v}
    \item The transition $(u,a,v')$ is inserted into $\delta_w$. \label{item: case 3 insert u a vp}
\end{enumerate}

\begin{description}
    \item[Invariant \ref{Loop Invariant: Deterministic}.] 
    The node $v'$ that is inserted by Action~\eqref{item: case 3 copy v} inherits its out-transitions from $v$ and hence determinism is maintained at $v'$.  Besides this change Action~\eqref{item: case 3 insert u a vp} inserts  $(u,a,v')$, but as Action~\eqref{item: case 3 remove uav} removes $(u, a, v)$ from $\Delta$, also this change maintains determinism of $\DDD_w^\Delta$.

    \item[Invariant \ref{Loop Invariant: Connectivity}.] 
    By the invariant, all states that exist before the actions are reachable from $q_0$ in  $\DDD_w$. These states are still reachable after the actions as no transitions are removed from $\delta_w$. The newly inserted state $v'$ is reachable from $u$ by the $(u, a, v')$ that is inserted into $\delta_w$ and $u$ is a state existing previous to the actions and thus reachable as argued before.
    
    \item[Invariant \ref{Loop Invariant: wheeler order}.] 
    Clearly, $q_0$ does not get a new incoming transition in $\delta_w$ as the only state that gets a new incoming transitions in $\delta_w$ is the new state $v'$. Furthermore, the new state $v'$ does get an incoming transition $(u, a, v')$ (Action~\eqref{item: case 3 insert u a vp}) and no other changes are made to $\delta_w$. We conclude that $q_0$ is still the only state without incoming transitions in $\DDD_w$.

    (W1) The actions do not remove any transitions and hence $\lambda_{\DDD_w}(x)$ does not change for any previously existing state $x$. The position $i$ in $\LEX$ of the new state $v'$ depends on whether $p=\bot$. If $p\neq \bot$, Action~\eqref{item: case 3 copy v} sets $i=1 + \LEX^{-1}[p]$ and thus $v'$ gets inserted immediately after $p$ with $\lambda_{\DDD_w}(p) = a$. If $p=\bot$, $i-1$ is the position of the last state in $\LEX$ with incoming letter strictly smaller than $a$ or $0$ if all states in $\LEX$ have an incoming letter at least $a$. We conclude that W1 is maintained by the actions. 

    (W2) The actions only touch transitions labeled with $a$, so we can ignore those labeled with other letters. There is a single newly inserted transition, namely Action~\eqref{item: case 3 insert u a vp} inserts the transition $(u, a, v')$ into $\delta_w$.    
    Let $(y, a, y')$ and $(z, a, z')$ be two transitions after the actions, with $y<z$. We need to prove that $y'\le z'$. There are two cases (I)~both $(y, a, y')$ and $(z, a, z')$ have already existed before the actions, (II)~exactly one out of the two transitions is the newly inserted transition $(u, a, v')$.

    In case (I), there is nothing to show as the relative order of $y,z,y',z'$ in $\LEX$ are not changed. 
    In case (II), let us assume that $(y, a, y')$ is newly inserted (the other case is symmetric), i.e., $y=u$, $y'=v'$, and $u<z$. We have to show that $v'\le z'$. As $(z, a, z')$ was a previously existing transition and $u<z$, it follows that $\succs(u, a)\neq\bot\neq s$. Moreover, by Action~\eqref{item: case 3 copy v}, it follows that $v'$ is inserted in $\LEX$ before $s$, i.e., $v'<s$. Furthermore, by definition $\succs(u,a)\le z$. 
    Now, if (1)~$\succs(u,a) = z$, then $s = z'$ (by determinism) and as $v'$ is inserted in $\LEX$ before $s$, $v' < s =z'$.
    If however (2)~$\succs(u,a) < z$, together with the existence of the two transitions $(\succs(u, a), a, s)$ and $(z, a, z')$ before the actions we get $s\le z'$. The fact that $v'$ is inserted in $\LEX$ before $s$ yields $v\le z'$.
    
    \item[Invariants \ref{Loop Invariant: MN-equiv} and \ref{Loop Invariant: language}.]
    Before the actions, by Invariant \ref{Loop Invariant: MN-equiv} it holds that $\equiv$ and $\cong$ are the same relation.
    Action~\eqref{item: case 3 copy v} creates a state $v'$ with the same outgoing transitions as $v$ and with its same ``final'' status. By Lemma~\ref{lem: modify}~\eqref{item: copy out-transitions}, 
    $v' \equiv v$ and the relation $\equiv$ remains unchanged, except for the addition of $v'$. 
    As Action~\eqref{item: case 3 copy v} sets $[v']_{\cong} = [v]_{\cong}$, this results in $\equiv$ and $\cong$ being identical again.
    Finally, by Lemma~\ref{lem: modify}~\eqref{item: move transition}, after the replacement of transition $(u,a,v)\in \Delta$ (by Action~\eqref{item: case 3 remove uav}) with transition $(u, a, v')\in \delta_w$ (through Action~\eqref{item: case 3 insert u a vp}), the equivalence relation $\equiv$ and the language $\LLL(\DDD_w^\Delta)$ do not change. This proves that both invariants are maintained. 
    
    \item[Invariant \ref{Loop Invariant: min order}.] 
    Through the actions, all previously existing states maintain their $\equiv$-equivalence class in $\DDD_w^\Delta$ (as established above) and their incoming labels in $\DDD_w$. Action~\eqref{item: case 3 copy v} inserts the new state $v'$ at position $i$, where $i$ is as described in the action above. We thus only need to check that the property holds (1) between positions $i-1$ and $i$ (for $i$ with $2\le i\le |\LEX|$) as well as (2) between positions $i$ and $i + 1$ (for $i$ with $1\le i\le |\LEX| - 1$). 
    We start with statement (1) and consider two cases depending on whether $p=\bot$. If $p\neq\bot$, we have that $\LEX[i - 1] = p$. The statement then follows from the fact that $[p]_{\cong} \neq [v]_{\cong} = [v']_{\cong}$ in this case and $\cong$ and $\equiv$ are the same equivalence relation (Invariant~\ref{Loop Invariant: MN-equiv}). Now assume that $p = \bot$. Then by Action~\eqref{item: case 3 copy v}, $i-1$ is the position of the last state in $\LEX$ with incoming letter strictly smaller than $a$ or $0$ if all states in $\LEX$ have an incoming letter at least $a$. As $i\ge 2$, the case $i-1 = 0$ cannot occur, otherwise we have $\lambda_{\DDD_w}(\LEX[i-1])\neq \lambda_{\DDD_w}(\LEX[i])$ and the statement holds.
    For~(2) we consider two cases depending on whether $s = \bot$. If $s\neq\bot$, we have that $\LEX[i + 1] = s$. The statement then follows from the fact that $[s]_{\cong} \neq [v]_{\cong}=[v']_{\cong}$ in this case and $\cong$ and $\equiv$ are the same equivalence relation (Invariant~\ref{Loop Invariant: MN-equiv}). Now assume that $s = \bot$. We are left to show that $\lambda_{\Delta_w}(\LEX[i + 1])\neq a$ in this case. We do so by arguing that $\LEX[i + 1]$ cannot have an in-transition on letter $a$. As $s=\bot$, we have $\succs(u, a)= \bot$ and thus there is no state $x$ with $u < x$ that has an outgoing transition with letter $a$. Hence an in-transition at $\LEX[i + 1]$ would have to come from a state $x$ with $x\le u$. Wheeler axiom (W2) (using Invariant~\ref{Loop Invariant: wheeler order}) and the existence of $(u, a, v')\in \delta_w$ imply that there is no state $x$ with $x<u$ such that $(x, a, \LEX[i + 1])$. Finally, due to determinism (Invariant~\ref{Loop Invariant: Deterministic}) there is no transition $(u, a, \LEX[i + 1])$. Hence $\lambda_{\Delta_w}(\LEX[i + 1])\neq a = \lambda_{\Delta_w}(\LEX[i])$ and this completes the proof.

    \item[Invariant~\ref{Loop Invariant: co-accessible}.] 
    Let $\DDD_w^\Delta$ be the DFA before the actions were executed. Consider now the DFA that results from $\DDD_w^\Delta$ with only Action~\eqref{item: case 3 copy v} executed (in particular this DFA still contains the transition $(u, a, v)$). As all out-transitions of $v$ in $\DDD_w^\Delta$ are copied as out-transitions to $v'$, it follows that $v'$ can reach the same final state in this DFA as $v$. Now consider the DFA after all actions, i.e., assume that in addition the replacement of transition $(u, a, v)$ with transition $(u, a, v')$ was executed through Actions~\eqref{item: case 3 remove uav}~and~\eqref{item: case 3 insert u a vp}. Note that we have $v\cong v'$ and by Invariant~\ref{Loop Invariant: MN-equiv} $v\equiv v'$. Lemma~\ref{lem: modify}~\eqref{item: move in transition} then implies that the DFA after these two actions is co-accessible.
\end{description}

\end{document}